\documentclass{svjour3}                     
\RequirePackage{fix-cm}
\smartqed  
\usepackage{mathptmx}      
\usepackage{amsmath,amssymb,graphicx}
\usepackage[hypertex]{hyperref}
\newcommand{\di}{\genfrac{}{}{0pt}{}}
\def\under#1{\kern.4pt\underline{\kern-.4pt{}#1\kern-.4pt}\kern.4pt}
\allowdisplaybreaks[4]
\journalname{preprint}
\dedication{Dedicated to Prof.\ Eberhard Zeidler on the occasion of
  his 75th birthday}

\begin{document}

\title{On the fixed point equation of a solvable 4D QFT model}

\author{Harald Grosse  \and  Raimar Wulkenhaar}

\institute{H. Grosse \at
              Fakult\"at f\"ur Physik, Universit\"at Wien, 
Boltzmanngasse 5, A-1090 Wien, Austria \\
              \email{harald.grosse@univie.ac.at}          
           \and
           R. Wulkenhaar \at
           Mathematisches Institut der Westf\"alischen
  Wilhelms-Universit\"at, 
Einsteinstra\ss{}e 62, D-48149 M\"unster, Germany\\
\email{raimar@math.uni-muenster.de}
}

\date{}


\maketitle

\begin{abstract}
The regularisation of the $\lambda\phi^4_4$-model on noncommutative
Moyal space gives rise to a solvable QFT model in which all correlation
functions are expressed in terms of the solution of a fixed point
problem. We prove that the non-linear operator for the logarithm of
the original problem satisfies the assumptions of the Schauder fixed
point theorem, thereby completing the solution of the QFT model.
\end{abstract}

\section{Introduction}
\label{intro}

This paper provides another key result in our long-term project on
quantum field theory on noncommutative geometries. This project was
strongly supported and influenced by Prof.\ Eberhard Zeidler.  One of
us (H.G.) spent a semester as Leibniz professor at the University of
Leipzig and enjoyed very much the hospitality at the
Max-Planck-Institute for Mathematics in the Sciences at
Inselstra{\ss}e, directed under Prof.\ Eberhard Zeidler. Shortly later
the other one of us (RW) was Schloe\ss{}mann fellow in the group of
Prof.\ Eberhard Zeidler. Our project started in this time.

\medskip

The first milestone was the
proof of perturbative renormalisability \cite{Grosse:2003aj},
\cite{Grosse:2004yu} of the $\lambda\phi^4_4$-model on Moyal space with
harmonic propagation. Eberhard Zeidler was constantly interested in
our work and played a decisive r\^ole in further development: He
understood that our computation of the $\beta$-function
\cite{Grosse:2004by} with the remarkable absence of the Landau ghost
problem \cite{Grosse:2004ik} could be of interest for Vincent
Rivasseau who visited the MPI Leipzig in summer 2004. Eberhard Zeidler
initiated a meeting of one of us (RW) with Vincent Rivasseau. This
contact led to a first joint publication \cite{Rivasseau:2005bh} which
brought the perturbative renormalisation proof of \cite{Grosse:2004yu}
closer to the constructive renormalisation programme
\cite{Rivasseau:1991ub}.  The growing group around Vincent Rivasseau
progressed much faster: they reproved the renormalisation theorem in
position space \cite{Gurau:2005gd}, derived the Symanzik polynomials
\cite{Gurau:2006yc}, extended the method to the Gross-Neveu model
\cite{VignesTourneret:2006nb} and so on \cite{Rivasseau:2007ab}.

The most important achievement started with a remarkable three-loop
computation of the $\beta$-function by Margherita Disertori and
Vincent Rivasseau \cite{Disertori:2006uy} in which they confirmed that
at a special self-duality point \cite{Langmann:2002cc}, the
$\beta$-function vanishes to three-loop order. Eventually, Margherita
Disertori, Razvan Gurau, Jacques Magnen and Vincent Rivasseau proved
in \cite{Disertori:2006nq} that the $\beta$-function vanishes to all
orders in perturbation theory. The central idea consists in combining 
the Ward identity for an $U(\infty)$ group action with Schwinger-Dyson 
equations.

We felt that the result of \cite{Disertori:2006nq} goes much deeper:  
Using these tools it must be possible to solve the model! Indeed we
succeeded in deriving a closed equation for the two-point function
of the self-dual model \cite{Grosse:2009pa}, which we 
renormalised and solved perturbatively to 3rd
order. The equation is a non-linear integral equation for a function
$G(\alpha,\beta)=:G_{\alpha\beta}$ on the unit square  $0\leq
\alpha,\beta< 1$:
\begin{align}
G_{\alpha\beta}
&=  1 
- \lambda \bigg( 
\frac{1-\alpha}{1-\alpha\beta} 
\big(\mathcal{M}_\beta-\mathcal{L}_\beta -\beta \mathcal{Y} \big)
+ \frac{1-\beta}{1-\alpha\beta} 
\big(\mathcal{M}_\alpha-\mathcal{L}_\alpha -\alpha \mathcal{Y}\big) 
\nonumber
\\
& 
+\frac{1-\beta}{1-\alpha\beta} 
\Big(\frac{G_{\alpha\beta}}{G_{0\alpha}} -1\Big)\big(
\mathcal{M}_\alpha-\mathcal{L}_\alpha + \alpha \mathcal{N}_{\alpha
  0}\big) 
-\frac{\alpha(1-\beta)}{1-\alpha\beta}  
\big(\mathcal{L}_{\beta}+\mathcal{N}_{\alpha\beta}-
\mathcal{N}_{\alpha 0}\big)
\nonumber
\\
&+  \frac{(1-\alpha)(1-\beta)}{1-\alpha \beta} 
(G_{\alpha\beta}-1)\mathcal{Y}
\bigg)\;,
\label{G-LMNY-neu}
\end{align}
where 
\begin{align}
\mathcal{L}_\alpha &:=\int_0^1 \!\! d\rho \; 
\frac{G_{\alpha\rho} -G_{0\rho}}{1-\rho} \;,
&
\mathcal{M}_\alpha &:=\int_0^1 \!\! d\rho \; 
\frac{\alpha\,G_{\alpha\rho}}{1-\alpha \rho} \;,
&
\mathcal{N}_{\alpha\beta} &:=\int_0^1 \!\! d\rho \; 
\frac{G_{\rho\beta}-G_{\alpha\beta}}{\rho -\alpha} \;,
\label{LMN}
\end{align}
and $\mathcal{Y}=\lim_{\alpha\to 0}
\frac{\mathcal{M}_\alpha-\mathcal{L}_\alpha}{\alpha}$. A solution
would be the key step to compute all higher correlation
functions. Unfortunately, all our attempts to solve this equation
failed, forcing us to put the problem aside for two years.

During the QFT workshop in November 2011 in Leipzig, one of us (RW) had the
chance to meet Eberhard Zeidler and to report about the programme: that
we succeeded to reduce all difficulties of a QFT model to a single
equation, but failed to solve it. Eberhard Zeidler immediately
offered help. He studied the problem (\ref{G-LMNY-neu})+(\ref{LMN}) during
the following three weeks, unfortunately without success.

This exchange led to a renewed interest and a subsequent major
breakthrough in spring 2012: We noticed that after suitable rescaling
of $G_{\alpha\beta}$ to $G_{ab}$, now with $a,b \in [0,\Lambda^2]$,
the difference function $D_{ab}=\frac{a}{b}(G_{ab}-G_{a0})$ satisfies
a \emph{linear} singular integral equation of Carleman type
\cite{Carleman} (the singular kernel is the $N_{\alpha\beta}$-integral
in (\ref{LMN})). We proved in \cite{Grosse:2012uv}, and with
corrections in \cite{Grosse:2014lxa} concerning a possible non-trivial
solution of the homogeneous Carleman equation \cite{Tricomi},
\cite{Muskhelishvili}, that given the boundary function $G_{a0}$ with
$G_{00}\equiv 1$, the full two-point function $G_{ab}$ reads
\begin{align}
G_{ab} &=
\frac{e^{\mathrm{sign}(\lambda)(\mathcal{H}_a^{\!\Lambda}[\tau_b]
-\mathcal{H}_0^{\!\Lambda}[\tau_0])} 
\sin (\tau_b(a))}{|\lambda| \pi a} \;,
\label{Gab}
\quad
\tau_b(a) :=\di{\raisebox{-1.2ex}{\mbox{\normalsize$\arctan$}}}{
\mbox{\scriptsize$[0,\pi]$}}
\Bigg(\dfrac{|\lambda| \pi a}{
b + \frac{1 +
\lambda \pi a \mathcal{H}_a^{\Lambda}[ G_{\bullet 0} ] }{G_{a 0}}}
\Bigg) \,.
\end{align}
By
$\mathcal{H}_a^{\!\Lambda}[f(\bullet)]:=\frac{1}{\pi}\lim_{\epsilon\to
  0} \Big(\int_0^{a-\epsilon}+\int_{a+\epsilon}^{\Lambda^2}
\Big)dx\;\frac{f(x)}{x-a}$ we denote the finite (or truncated) Hilbert
transform. We are mainly interested in the one-sided Hilbert transform
$\mathcal{H}_a^{\!\infty}[f(\bullet)]:=\lim_{\Lambda^2\to \infty}
\mathcal{H}_a^{\!\Lambda}[f(\bullet)]$.
As shown in \cite{Grosse:2014lxa}, this result is correct 
for $\lambda<0$, which is the interesting case
for reflection positivity \cite{Grosse:2013iva}.
For $\lambda>0$ one has to multiply 
(\ref{Gab}) by a factor $(1+\frac{\Lambda^2}{\Lambda^2-a}
(aC+F(b)))$, where $C$ is a constant and $F(b)$ an arbitrary
function with $F(0)=0$.

The symmetry condition $G_{ab}=G_{ba}$ of a two-point function 
leads for $a=0$ and $\lambda<0$ to the consistency condition
(in the limit $\Lambda\to \infty$)
\begin{align}
G_{b0}&= G_{0b}=
\frac{1}{1+b}
\exp\Bigg(
{-} \lambda 
\int_0^b \!\!\! dt  \int_0^{\infty}  \!\!\!
\frac{dp}{(\lambda \pi p)^2 
+\big( t + \frac{1 +\lambda \pi p 
 \mathcal{H}_p^{\infty}[G_{\bullet 0}]}{G_{p0}}\big)^2} 
\Bigg)\;.
\label{G0b}
\end{align}
Equation (\ref{G0b}) is a much simpler problem than
(\ref{G-LMNY-neu})+(\ref{LMN}). In \cite{Grosse:2012uv} we already proved
existence of a solution for $\lambda>0$ via the Schauder fixed point
theorem.  This case turned out to be much less interesting than
$\lambda<0$: Reflection positivity is excluded for $\lambda>0$
\cite{Grosse:2013iva}, and the formulae (\ref{Gab})+(\ref{G0b}) need to be
corrected by a winding number \cite{Grosse:2014lxa}.

The proof for $\lambda>0$ given in \cite{Grosse:2012uv} does not
generalise to the opposite sign.  In this paper we fill the gap and
prove that (\ref{G0b}) has a solution for $-\frac{1}{6}\leq
\lambda<0$. The key is to focus on the logarithm of $G_{a0}$, which is
an unbounded function. We are able to control the divergence at
$\infty$ and prove uniform continuity of the Hilbert transform on such
spaces.  For $-\frac{1}{6}\leq \lambda\leq 0$ we are able to verify the
assumptions of the \emph{Schauder fixed point theorem} 
so that (\ref{G0b}) has a solution with good
additional properties. We would like to warn the reader that the
estimates are cumbersome.

The Schauder fixed point theorem is a central topic in 
Eberhard Zeidler's book \cite[Chap.\ 2]{Zeidler:1986??}. It follows
from Brouwer's fixed point theorem for which an elementary proof
is given in \cite[\S 77]{Zeidler:1988??}.

\medskip

It is a pleasure to dedicate this paper to Prof.\ Eberhard Zeidler who showed
constant interest in our programme and provided strategic help. 
From our early common interaction on we were
strongly supported by the MPI (and ESI in Vienna), which allowed
  our long-standing fruitful interaction. We congratulate
  Prof.\ Zeidler to his birthday and wish him many happy recurrences.
We
hope he enjoys the connection between quantum field
theory \cite{Zeidler:2006rw} \cite{Zeidler:2009zz} \cite{Zeidler:2011zz}
and non-linear functional analysis 
\cite{Zeidler:1986??} \cite{Zeidler:1990??}, \cite{Zeidler:1985??}, 
\cite{Zeidler:1988??}.

\section{Logarithmically bounded functions}

Consider the following vector space of real-valued functions
\begin{align}
LB:=\Big\{f\in \mathcal{C}^1(\mathbb{R}_+)\;:~ 
f(0)=0\;,~ |f'(x)| \leq  \frac{C}{1+x}\text{ for some } C\geq 0\Big\}\;.
\label{LB}
\end{align}
These functions vanish at zero and grow/decrease at most
logarithmically at $\infty$. We equip $LB$ with the norm 
\begin{align}
\|f\|_{LB}:=|f(0)|+\sup_{x\geq 0} \big|(1+x) f'(x)\big|
\qquad \text{for } f\in LB\;.
\label{norm-LB}
\end{align}
Indeed, $\|f\|_{LB}=0$ means $f(0)=0$ and $|f'|=0$, hence $f'=0$ and
thus $f(x)=0$ everywhere. The addional $|f(0)|$ is redundant but makes
it easier to formulate the proofs.
\begin{proposition}
$(LB,\|~\|_{LB})$ is a Banach space.
\end{proposition}
\begin{proof}
Given a Cauchy sequence $(f_n)_{n\in \mathbb{N}}$ in $LB$. This means 
$f_n(0)=0$ for every $n$, and for every $\epsilon>0$
there is $N_\epsilon\in \mathbb{N}$ with $\|f_n-f_m\|_{LB}=\sup_{x\geq 0}
|(1+x)f'_n(x)-(1+x)f'_m(x)|
<\epsilon$ for all $m,n\geq N_\epsilon$. This implies
$|(1+x) f'_n(x)-(1+x)f'_m(x)|<\epsilon$ for every $x\geq 0$.
By the completeness of $\mathbb{R}$, the sequence
$\big((1+x)f'_n(x)\big)_{n\in \mathbb{N}}$ converges at
every $x\geq 0$ and defines a limit function 
$(1+x)g(x):=\lim_{n\to \infty} (1+x)f'_n(x)$. 
Taking the limit $m\to \infty$ above shows that 
\begin{align}
\big|(1+x)f'_n(x)-(1+x)g(x)\big|<\epsilon\qquad 
\text{for every $x$ and $n\geq N_\epsilon$}. 
\tag{*}
\end{align}
Fix such $n\geq N_\epsilon$. By definition of $f_n\in LB$, the 
derivative $x\mapsto (1+x)f_n'(x)$ is continuous at
every $x$. This means that there is $\delta_x>0$ such that
$|(1+x)f'_n(x)-(1+y)f'_n(y)|< \epsilon$ for all $y\geq 0$ with
$|x-y|<\delta_x$. For such $y$ it follows
\begin{align*}
\big|(1{+}x)g(x)-(1{+}y)g(y)\big| &\leq 
\big|(1+x)g(x)-(1+x)f_n'(x)\big| 
+\big|(1+x) f_n'(x)-(1+y)f'_n(x)\big| 
\\
&+\big|(1+y) f_n'(y)-(1+y)g(y)\big| < 3\epsilon\;.
\end{align*}
Therefore, the limit function $t\mapsto (1+t)g(t)$ and hence 
$t\mapsto g(t)$ is continuous. As such it can be
integrated over any compact interval. We \emph{define} a function
$f(x)$ by
\[
f(x)=\int_0^x dt \,g(t)\;.
\]
This means $f(0)=0$, and by the fundamental theorem of calculus the
function $f$ is differentiable at every $x\geq 0$, and $f'(x)=g(x)$ is
continuous.
Expressing this as $(1+x)g(x)=(1+x)f'(x)$ we have proved with (*)
\begin{align*}
|f_n(0)-f(0)|=0 \;,\qquad
|(1+x)f'_n(x)-(1+x)f'(x)|<\epsilon\quad \text{for every $x$ 
and $n\geq N_\epsilon$}\;. 
\end{align*}
Hence, $(f_n)_{n\in \mathbb{N}}$ converges to a function $f \in
\mathcal{C}^1(\mathbb{R}_+)$ in the $LB$-norm. By construction 
we have $f\in LB$, hence $(LB,\|~\|_{LB})$ is complete. \hfill $\square$%
\end{proof}

Consider for $-\frac{1}{3}<\lambda<0$ the following subset 
\begin{align}
\mathcal{K}_\lambda=\Big\{f\in LB\;:~ f(0)=0\;,\qquad -\frac{1-|\lambda|}{1+x}
\leq f'(x) \leq 
-\frac{1-\frac{|\lambda|}{1-2|\lambda|}}{1+x}  \Big\}\subseteq LB\;.
\label{calK}
\end{align}
\begin{lemma}
$\mathcal{K}_\lambda$ is a norm-closed subset of the Banach space $LB$.
\end{lemma}
\begin{proof} 
The evaluation maps $\widetilde{ev},ev_x :LB \to \mathbb{R}$, with
$\widetilde{ev}(f)=f(0)$ and $ev_x(f)=(1+x)f'(x)$ are continuous 
maps from $LB$ to $\mathbb{R}$. Hence, the following subset is closed
in $LB$:
\begin{align}
\mathcal{K}_\lambda= \widetilde{ev}^{-1}(\{0\}) \cap 
\bigcap_{x\geq 0} ev_x^{-1}\Big(\big[
-( 1-|\lambda|),-( 1-\tfrac{|\lambda|}{1-2|\lambda|})\big]\Big)\;.
\tag*{\mbox{$\square$}}
\end{align}
\end{proof}

In the sequel we use implicity the fact that the Hilbert 
transform of a function that simultaneously belongs for some $p>1$ to 
$L^p({[\Lambda^2,\infty[}$ and to the
$\alpha$-H\"older space on ${]0,\Lambda^2[}$ for some $0<\alpha<1$ is again 
a H\"older-continues function with the same H\"older exponent
$\alpha$. For functions on ${]{-}\pi,\pi[}$ this was proved by Priwaloff
\cite{Priwaloff:1916} for a variant of the Hilbert 
transform. This proof is easily generalised to ${]0,\Lambda^2[}$. The
$L^p$ condition is necessary for Hilbert transforms over $\mathbb{R}$ 
and clearly extends to the one-sided Hilbert transform over 
$\mathbb{R}_+$.  

This means that for $f\in \mathcal{K}_\lambda$ the following maps are
well defined (possibly with integrals restricted to
$[\epsilon,\Lambda^2]$; the convergence on $\mathbb{R}_+$ will be
verified in the following section):
\begin{subequations}
\begin{align}
Rf(a) &:= \frac{1-|\lambda|\pi a \mathcal{H}_a^{\!\infty}
[e^{f(\bullet)}]}{e^{f(a)}}\;,
\label{Rf}
\\
Tf(b) &:= 
-\log (1{+}b) 
+ \int_0^{\infty}\!
\frac{dt}{\pi t}\Big(
\arctan \frac{b+Rf(t)}{|\lambda|\pi t}-
\arctan \frac{Rf(t)}{|\lambda|\pi t}
\Big)\;.
\label{Tfa}
\end{align}
\end{subequations}
Formula (\ref{Tfa}) involves the standard branch of the
$\arctan$-function with range ${]{-}\frac{\pi}{2},\frac{\pi}{2}[}$,
related to the branch used in (\ref{Gab}) by 
$\di{\raisebox{-1.2ex}{\mbox{\normalsize$\arctan$}}}{
\mbox{\scriptsize$[0,\pi]$}} (x) =\frac{\pi}{2}-\arctan \frac{1}{x}$.
Comparing with (\ref{Gab}) at $a=0$, equivalent to (\ref{G0b}), shows 
$\log G_{0b}=(T \log G_{\bullet 0})(b)$. 

In the following three sections we prove three main results (for a
restricted set of $|\lambda|$): that $T$ maps $\mathcal{K}_\lambda$
into itself, that $T$ is norm-continuous on $\mathcal{K}_\lambda$ and
that the image $T\mathcal{K}_\lambda\subseteq \mathcal{K}_\lambda$ is
relatively compact:
\begin{theorem}
For $-\frac{1}{6}\leq \lambda \leq 0$,
consider the map $T$ defined by (\ref{Tfa}) on the subset 
$\mathcal{K}_\lambda \subseteq LB$ of the Banach space of 
logarithmically bounded function, see
(\ref{LB}), (\ref{norm-LB}) and (\ref{calK}).
Then for any $f\in \mathcal{K}_\lambda$ one has
\begin{enumerate}

\item[i)]  $Tf \in \mathcal{K}_\lambda$.

\item[ii)] $T:\mathcal{K}_\lambda \to \mathcal{K}_\lambda$ is norm-continuous.

\item[iii)] The restriction of $T\mathcal{K}_\lambda$ to any interval
  $[0,\Lambda^2]$ is relatively compact in  norm-topology. 

\end{enumerate}
In particular,
$T$ has a fixed point $f_*=Tf_* \in \mathcal{K}\big|_{[0,\Lambda^2]}$ which
we denote $\log G_{0b}:=f_*(b)$.
\label{maintheorem}
\end{theorem}
\begin{proof}
The domain $\mathcal{K}$ is also convex. Then i),ii),iii) are the
requirements of the Schauder fixed point theorem \cite[Chapter
2]{Zeidler:1986??} to guarantee
existence of fixed point $Tf_*=f_*$. 
The proof of i),ii),iii) is given in the 
following subsections. \hfill $\square$%
\end{proof}

In this way we prove existence of function $G_{0b}=G_{b0}$ which satisfies 
(\ref{G0b}) for all $0\leq b\leq \Lambda$. For $b>\Lambda^2$ 
there is possibly a discrepancy. Since both sides of 
(\ref{G0b})  belong to $\mathcal{K}_\lambda$ the error is $\leq 
(1+\Lambda^2)^{\frac{|\lambda|}{1-2|\lambda|}-1}-
(1+\Lambda^2)^{|\lambda|-1}$. To put it differently, for every $\epsilon>0$ 
there is $G_{0b} \in \exp \mathcal{K}_\lambda$ such that the difference between 
lhs and rhs of (\ref{G0b}), and consequently also the difference 
between their derivatives,  is $<\epsilon$. This statement means 
that (\ref{G0b}) has a solution in $\mathcal{C}^1_0(\mathbb{R}_+)$.

\section{$T$ preserves $\mathcal{K}_\lambda$}

\label{sec:TKK}

Integrating the definition (\ref{calK}) of $\mathcal{K}_\lambda$ from 
$a$ to $x>a$ yields
\[
\log\Big(\frac{1+a}{1+x}\Big)^{1-|\lambda|}
\leq f(x)-f(a) \leq \log\Big(\frac{1+a}{1+x}\Big)^{1-\frac{|\lambda|}{1-2|\lambda|}}
\]
and consequently (for $x>a$)
\begin{align}
\Big(\frac{1+a}{1+x}\Big)^{1-\lambda}
\leq 
\frac{e^{f(x)}}{e^{f(a)}} \leq  
\Big(\frac{1+a}{1+x}\Big)^{1-\frac{|\lambda|}{1-2|\lambda|}},\quad
\Big(\frac{1+x}{1+a}\Big)^{1-\frac{|\lambda|}{1-2|\lambda|}}
\leq \frac{e^{f(a)}}{e^{f(x)}} \leq  
\Big(\frac{1+x}{1+a}\Big)^{1-\lambda},
\end{align}
which we reinterpet as
\begin{align}
\left.
\begin{array}{c} 
\displaystyle  \Big(\frac{1+a}{1+x}\Big)^{1-\lambda}
\\
\displaystyle 
\Big(\frac{1+a}{1+x}\Big)^{1-\frac{|\lambda|}{1-2|\lambda|}}
\end{array}
\right\} 
\leq 
\frac{e^{f(x)}}{e^{f(a)}} \leq  \left\{
\begin{array}{cl} 
\displaystyle \Big(\frac{1+a}{1+x}\Big)^{1-\frac{|\lambda|}{1-2|\lambda|}}
& \quad \text{for } x>a \;,
\\
\displaystyle 
\Big(\frac{1+a}{1+x}\Big)^{1-\lambda}
& \quad \text{for } x<a\;.
\end{array}
\right.
\label{efg-5}
\end{align}
We take the one-sided Hilbert transform:
\begin{align} 
\frac{\mathcal{H}^{\!\infty}_a[e^{f(\bullet)}]}{e^{f(a)}}
&= \frac{1}{\pi} 
\lim_{\epsilon\to 0} \Big\{
-\int_0^{a-\epsilon} \!\! \frac{dx}{(a-x)} 
\frac{e^{f(x)}}{e^{f(a)}} + 
\int_{a+\epsilon}^\infty \frac{dx}{(x-a)} 
\frac{e^{f(x)}}{e^{f(a)}} \Big\}\;.
\label{Hilbert-ef}
\end{align}

The Hilbert transform (\ref{Hilbert-ef}) becomes maximal if for $x>a$ 
we use the maximal 
$\frac{e^{f(x)}}{e^{f(a)}}$ but for $x<a$ the minimal 
$\frac{e^{f(x)}}{e^{f(a)}}$. Conversely,
the Hilbert transform becomes minimal 
if for $x>a$ 
we use the minimal $\frac{e^{f(x)}}{e^{f(a)}}$ but for $x<a$ the
maximal
$\frac{e^{f(x)}}{e^{f(a)}}$:
\begin{align} 
&\frac{1}{\pi} \lim_{\epsilon\to 0}
\Big\{
-\int_0^{a-\epsilon} \frac{dx\;(1+a)^{1-\lambda} }{(a-x)(1+x)^{1-\lambda}}
+ 
\int_{a+\epsilon}^\infty \frac{dx\;(1+a)^{1-\lambda} }{(x-a)(1+x)^{1-\lambda}}\Big\}
\nonumber
\\
&\leq 
\frac{\mathcal{H}^{\!\infty}_a[e^{f(\bullet)}]}{e^{f(a)}} 
\leq  
\frac{1}{\pi} \lim_{\epsilon\to 0}
\Big\{
{-}\int_0^{a-\epsilon} \!\!\! \frac{dx\;(1{+}a)^{1-\frac{|\lambda|}{1-2|\lambda|}} 
}{(a{-}x)(1{+}x)^{1-\frac{|\lambda|}{1-2|\lambda|}}} 
+ 
\int_{a+\epsilon}^\infty \!\frac{dx\;(1+a)^{1-\frac{|\lambda|}{1-2|\lambda|}} 
}{(x{-}a)(1{+}x)^{1-\frac{|\lambda|}{1-2|\lambda|}} }\Big\}
\;.
\label{Hilbert-minmax-0}
\end{align}
Note that the analogue only for $\mathcal{H}^{\!\infty}_a[e^{f(\bullet)}]$ would not
hold; in that case the opposite boundaries of $\mathcal{K}_\lambda$ would
contribute to $x<a$ versus $x>a$, and there is no chance of a
reasonable estimate!
We can reformulate (\ref{Hilbert-minmax-0}) as
\begin{align} 
\frac{\mathcal{H}^{\!\infty}_a\big[(1+\bullet)^{|\lambda|-1}\big]}{
(1+a)^{|\lambda|-1}}
&\leq 
\frac{\mathcal{H}^{\!\infty}_a[e^{f(\bullet)}]}{e^{f(a)}}
\leq  
\frac{\mathcal{H}^{\!\infty}_a\big[(1+\bullet)^{\frac{|\lambda|}{1-2|\lambda|}-1}
\big]}{(1+a)^{\frac{|\lambda|}{1-2|\lambda|}-1}}\;.
\label{Hilbert-minmax}
\end{align}
We prove the following result which covers a slightly more general
case:
\begin{proposition}
For any $\mu<1$, with $\mu\neq 0$, and $\beta>0$ one has
\begin{align}
\frac{\mathcal{H}_a^{\infty}\big[(\beta+\bullet)^{\mu-1}\big]}{
(\beta+a)^{\mu-1}}
= - \cot (\pi \mu) 
+\frac{1}{\mu\pi}  \Big(\frac{\beta}{\beta+a}\Big)^\mu 
{}_2F_1\Big(\di{1,\mu}{1+\mu}\Big| \frac{\beta}{a+\beta}\Big)\;.
\label{Hilbert-2F1}
\end{align}
\end{proposition}
\begin{proof}
We use the following indefinite integrals:
\begin{subequations}
\begin{align}
\int dx \frac{(\beta +x)^{\mu-1}}{x+c}&=-\frac{(\beta+x)^{\mu-1}}{1-\mu} \;
{}_2F_1\Big(\di{1,1-\mu}{2-\mu}\Big|
\frac{-c+\beta}{x+\beta}\Big)\;, 
&x&>-c\;,
\label{HyperG+}
\\
\int dx \frac{(\beta +x)^{\mu-1}}{a-x}&=\frac{(\beta+x)^{\mu-1}}{\mu} 
\frac{\beta+x}{\beta+a}
\;
{}_2F_1\Big(\di{1,\mu}{1+\mu}\Big| \frac{x+\beta}{a+\beta}\Big)\;,
 &x &< a\;.
\label{HyperG-}
\end{align}
\end{subequations}
This is proved via $x$-differentiation using 
$\frac{d}{dx}{}_2F_1\big(\di{\alpha,\;\beta}{\gamma}\big|x\big)=
\frac{\alpha\beta}{\gamma}
{}_2F_1\big(\di{\alpha+1,\;\beta+1}{\gamma+1}\big|x\big)$
and use of the recursion relations \cite[\S 9.137]{Gradsteyn:1994??} for 
the hypergeometric function.
With a large cut-off $\Lambda^2$ we have for $\mu<1$ 
\begin{subequations}
\label{Hilbert-proof}
\begin{align} 
&\pi \mathcal{H}_a^{\infty}\big[(\beta+\bullet)^{\mu-1}\big]
= 
\lim_{\epsilon\to 0,\Lambda^2\to \infty}
\Big\{
-\int_0^{a-\epsilon} \!\!\! dx \;\frac{(\beta {+}x)^{\mu-1}}{a-x}
+
\int_{\epsilon}^{\Lambda^2-a} \!\!\! dx
\;\frac{(\beta {+}a{+}x)^{\mu-1}}{x}
\Big\}
\nonumber
\\
&= \lim_{\epsilon\to 0,\Lambda^2\to \infty}\Big\{ 
-\frac{(\beta+x)^{\mu-1}}{\mu} \frac{(\beta+x)}{(\beta+a)}
{}_2F_1\Big(\di{1,\mu}{1+\mu}\Big| \frac{\beta+x}{\beta+a}\Big)
\Big|_0^{a-\epsilon}
\nonumber
\\
&\qquad 
-\frac{(\beta+a+x)^{\mu-1}}{1-\mu} 
{}_2F_1\Big(\di{1,1-\mu}{2-\mu}\Big| \frac{\beta+a}{\beta+a+x}\Big)
\Big|_{\epsilon}^{\Lambda^2-a} \Big\}
\nonumber
\\
&= 
\frac{\beta^\mu}{\mu(\beta+a)}
{}_2F_1\Big(\di{1,\mu}{1+\mu}\Big| \frac{\beta}{\beta+a}\Big)
\label{Ponnusamy-final}
\\
&+ \lim_{\epsilon\to 0}
\Big\{
- \frac{(\beta+a-\epsilon)^\mu}{\beta+a}\,
B(1,\mu)\, {}_2F_1\Big(\di{1,\mu}{1+\mu}\Big| 
\frac{\beta+a-\epsilon}{\beta+a}\Big)
\nonumber
\\
& \qquad\qquad +(\beta+a+\epsilon)^{\mu-1} B(1,1-\mu)
{}_2F_1\Big(\di{1,1-\mu}{2-\mu}\Big| \frac{\beta+a}{\beta+a+\epsilon}\Big)
\Big\}\;,
\label{Ponnusamy-eps}
\end{align}
\end{subequations}
where the special values $B(1,1-\mu)=\frac{1}{1-\mu}$ and 
$B(1,\mu)=\frac{1}{\mu}$ for the Beta function have been used.
The limit $\epsilon\to 0$ is controlled by the following result in
\cite{Ponnusamy:1997??} (already claimed, but not proved, in Ramanujan's
notebooks) for zero-balanced hypergeometric functions:
If $0<\alpha,\beta,x \leq 1$, then
\begin{align}
-\psi(\alpha)-\psi(\beta)-2\gamma&<
B(\alpha,\beta)
{}_2F_1\Big(\di{\alpha,\;\beta}{\alpha+\beta}\Big|1-x\Big)
+\log(x)
\nonumber
\\*
&
< -\psi(\alpha)-\psi(\beta)-2\gamma+\frac{x}{1-x} \log \frac{1}{x}\;.
\label{Ponnusamy}
\end{align}
Here $\psi(x)=\frac{\Gamma'(x)}{\Gamma(x)}$, and $\gamma=-\psi(1)$ is the
Euler-Mascheroni constant. Since 
\begin{align*}
\lim_{\epsilon\to 0}
\Big\{ 
-\frac{(\beta+a-\epsilon)^\mu}{\beta+a} 
\log \Big(\frac{\epsilon}{\beta+a}\Big)
+\frac{(\beta+a+\epsilon)^\mu}{(\beta+a+\epsilon)}
\log \Big(\frac{\epsilon}{\beta+a+\epsilon}\Big)\Big\}=0\;,
\end{align*}
we can add the corresponding $\log$-terms to 
(\ref{Ponnusamy-eps}) and use (\ref{Ponnusamy}) to conclude that the
two lines (\ref{Ponnusamy-eps}) converge in the 
limit $\epsilon\to 0$ to 
\begin{align}
\lim_{\epsilon\to 0} \textup{(\ref{Ponnusamy-eps})}
= (\beta+a)^{\mu-1} \big( \psi(\mu)-\psi(1-\mu)\big)
= -(\beta+a)^{\mu-1} \pi \cot (\pi \mu)\;,
\tag{\ref{Hilbert-proof}c}
\label{Ponnusamy-eps-prime}
\end{align}
where \cite[\S 8.365.8]{Gradsteyn:1994??} has been used.
This finishes the proof.
\hfill $\square$%
\end{proof}

Inserting (\ref{Hilbert-2F1}) for $\beta=1$ and
$\mu=|\lambda|,\frac{|\lambda|}{1-2|\lambda|}$, respectively, 
into (\ref{Hilbert-minmax}) gives the following
bounds valid for any $f\in \mathcal{K}_\lambda$: 
\begin{align}
&-  \cot (|\lambda|\pi) +\frac{1}{|\lambda|\pi (1+a)^{|\lambda|}}\;
{}_2F_1\Big(\di{1,|\lambda|}{1+|\lambda|}\Big| \frac{1}{1+a}\Big)
\nonumber
\\*
&\leq
\frac{\mathcal{H}^{\!\infty}_a[e^{f(\bullet)}]}{e^{f(a)}}
\leq 
-  \cot \Big(\frac{|\lambda|\pi}{1-2|\lambda|}\Big)
+\frac{1-2|\lambda|}{|\lambda|\pi (1+a)^{\frac{|\lambda|}{1-2|\lambda|}}}
\;{}_2F_1\Big(\di{1,\frac{|\lambda|}{1-2|\lambda|}}{1+\frac{|\lambda|}{1-2|\lambda|}
}\Big| \frac{1}{1+a}\Big)\;.
\end{align}
Together with (\ref{efg-5}) taken at $x=0$ we
obtain for the function $Rf$ defined in (\ref{Rf}) the following
bounds:
\begin{align}
&|\lambda|\pi a \cot \Big(\frac{|\lambda|\pi}{1-2|\lambda|}\Big)
+
\frac{1+a}{(1+a)^{\frac{|\lambda|}{1-2|\lambda|}}}
-\frac{(1-2|\lambda|)a}{(1+a)^{\frac{|\lambda|}{1-2|\lambda|}}}
{}_2F_1\Big(\di{1,\frac{|\lambda|}{1-2|\lambda|}}{
1+\frac{|\lambda|}{1-2|\lambda|}
}\Big| \frac{1}{1+a}\Big)
\nonumber
\\
&\leq (Rf)(a)
\leq 
|\lambda|\pi a  \cot (|\lambda|\pi) 
+\frac{1+a}{(1+a)^{|\lambda|}}
-\frac{a}{(1+a)^{|\lambda|}}
{}_2F_1\Big(\di{1,|\lambda|}{1+|\lambda|}\Big| \frac{1}{1+a}\Big)\;.
\label{Rfbound1}
\end{align}
Since 
${}_2F_1\big(\di{1,|\lambda|}{1+|\lambda|}\Big| \frac{1}{1+a}\big)\geq
1$ we have 
$\frac{a}{(1+a)^{|\lambda|}}
-\frac{a}{(1+a)^{|\lambda|}}
{}_2F_1\big(\di{1,|\lambda|}{1+|\lambda|}\Big| \frac{1}{1+a}\big)
\leq 0$. This means that the upper bound is smaller than 
$|\lambda|\pi a  \cot (|\lambda|\pi) 
+\frac{1}{(1+a)^{|\lambda|}}\leq 
|\lambda|\pi a  \cot (|\lambda|\pi)+1$.
In the lower bound we use \cite[\S 9.137.12]{Gradsteyn:1994??} to 
write the hypergeometric function as 
${}_2F_1\Big(\di{1,\frac{|\lambda|}{1-2|\lambda|}}{
1+\frac{|\lambda|}{1-2|\lambda|}
}\Big| \frac{1}{1+a}\Big)
=1+\frac{|\lambda|}{(1-|\lambda|)(1+a)}\;
{}_2F_1\Big(\di{1,1+\frac{|\lambda|}{1-2|\lambda|}}{
2+\frac{|\lambda|}{1-2|\lambda|}
}\Big| \frac{1}{1+a}\Big)$. This gives,
partly expressed in terms of 
$|\lambda_r|:=\frac{|\lambda|}{1-2|\lambda|}$,  
\begin{align}
&(|\lambda|\pi a) \cot \Big(\frac{|\lambda|\pi}{1-2|\lambda|}\Big)
+ 1 + |\lambda| F_{\lambda_r}(a) 
\leq (Rf)(a) \leq |\lambda|\pi a  \cot (|\lambda|\pi) +1\;,
\quad 
\text{where} \nonumber \\
&
F_{\lambda_r}(a):=
\frac{1+2|\lambda_r|}{|\lambda_r|}
\Big(\frac{1+|\lambda_r|a }{(1+a)^{|\lambda_r|}}-1\Big)
+\frac{\hat{F}_{\lambda_r}(a)}{(1+a)^{|\lambda_r|}}\;,
\label{Flra}
\\
& 
\hat{F}_{\lambda_r}(a)
:= 
(1-2|\lambda_r|)a
-\frac{a}{(1+|\lambda_r|)(1+a)}\;
{}_2F_1\Big(\di{1,1+|\lambda_r|}{2+|\lambda_r|}
\Big| \frac{1}{1+a}\Big)\;.
\nonumber
\end{align}
We have to show
that $F_{\lambda_r}(a)$ is of positive mean for a certain integral. This
is easy to check for a computer, but we want to make it rigorous. 
For a lower bound we can remove the numerator $(1+2|\lambda_r|)$ in 
the middle line 
of (\ref{Flra}). The remaining piece 
$\frac{1}{|\lambda_r|}
\Big(\frac{1+|\lambda_r|a }{(1+a)^{|\lambda_r|}}-1\Big)$ is positive 
for $0 < |\lambda_r| <1$ by a particular case of Bernoulli's inequality. 
Then its $|\lambda_r|$-derivative reads
\[
\frac{d}{d|\lambda_r|}
\Big(\frac{1}{|\lambda_r|}
\Big(\frac{1+|\lambda_r|a }{(1+a)^{|\lambda_r|}}-1\Big)\Big)
=\frac{-1 + (1+a)^{|\lambda_r|}
-
(1+|\lambda_r|a) \log((1+a)^{|\lambda_r| }) }{
|\lambda_r|^2(1+a)^{|\lambda_r|}}\;.
\]
Using again Bernoulli's inequality, the  numerator is 
$\leq x-(1+x)\log(1+x)$ with $x:= (1+a)^{|\lambda_r|}-1$. The function 
$x-(1+x)\log(1+x)$ vanishes at $x=0$ and has negative derivative 
for any $x>0$. Consequently, 
$\frac{1}{|\lambda_r|}
\big(\frac{1+|\lambda_r|a }{(1+a)^{|\lambda_r|}}-1\big)$ is 
monotonously decreasing in $|\lambda_r|$ (hence in $|\lambda|)$ for 
any fixed $a$. 

We expand $\hat{F}_{\lambda_r}(a)$ in the last line of (\ref{Flra}) 
into a power series and take the 
$|\lambda_r|$-derivative:
\begin{align*}
\frac{d}{d|\lambda_r|}\hat{F}_{\lambda_r}(a)
&=-2a
+ a \sum_{k=0}^\infty \frac{1}{(k+1+|\lambda_r|)^2}
\frac{1}{(1+a)^{k+1}} < a\Big(-2+\frac{\pi^2}{6}\Big)\;,
\\
\frac{d}{d|\lambda_r|}\Big(\frac{\hat{F}_{\lambda_r}(a)}{(1+a)^{|\lambda_r|}}\Big)
&< \frac{a}{(1+a)^{|\lambda_r|}}
\Big(-2+\frac{\pi^2}{6}
- \frac{\log(1+a)}{a}\hat{F}_{\lambda_r}(a)\Big)\;.
\end{align*}
Hence also $\hat{F}_{\lambda_r}(a)$ is decreasing in $|\lambda_r|$,
and sufficient for extending this decrease to $F_{\lambda_r}(a)$ is 
$\hat{F}_{\lambda_r}(a)\geq -(2-\frac{\pi^2}{6})$. Using identities and recursion formulae such as \cite[\S 9.137.14+17
\S 9.131.1]{Gradsteyn:1994??} for the hypergeometric function it is
straightforward to compute and rearrange the derivatives of
$\hat{F}_\lambda$:
\begin{subequations}
\begin{align}
\hat{F}'_{\lambda_r}(a)
&=
1-2|\lambda_r| 
-\frac{1}{(1+a)^2}
\frac{1}{(1+\lambda_r|)(2+|\lambda_r|)}\;
{}_2F_1\Big(\di{2,1+|\lambda_r|}{3+|\lambda_r|}
\Big| \frac{1}{1+a}\Big)\;,
\label{F-prime}
\\
\hat{F}''_{\lambda_r}(a)
&=
\frac{2}{a(1+a)^2}
\frac{1}{(1+|\lambda_r|)(2+|\lambda_r|)}\;
{}_2F_1\Big(\di{2,|\lambda_r|}{3+|\lambda|_r}
\Big| \frac{1}{1+a}\Big)\;.
\label{F-primeprime}
\end{align}
\end{subequations}
From 
(\ref{F-primeprime}) we conclude that $\hat{F}$ is convex in $a$ for
any fixed $|\lambda_r|$, and  
(\ref{F-prime}) shows that $\hat{F}$ starts negative near $a=0$ and 
diverges (in case of $|\lambda_r|<\frac{1}{2}$) 
to $+\infty$ for $a\to \infty$.
Together with convexity, there is a 
unique zero $\hat{F}_{\lambda_r}(t_\lambda)=0$ at $t_\lambda>0$ and
a single and unique global=local minimum in $[0,t_\lambda]$.
One can check numerically or by
estimating the power series that $\hat{F}_{-\frac{1}{4}}(\frac{3}{2})>0$ and
$\hat{F}'_{-\frac{1}{4}}(\frac{1}{5})<0$. By convexity,
$\hat{F}_{-\frac{1}{4}}$ lies above any tangent, and the intersection
of the tangent $\hat{F}_{-\frac{1}{4}}(\frac{1}{5})+(t-\frac{1}{5}) 
\hat{F}'_{-\frac{1}{4}}(\frac{1}{5})$ with the tangent 
$\hat{F}_{-\frac{1}{4}}(\frac{3}{2})+(t-\frac{3}{2}) 
\hat{F}'_{-\frac{1}{4}}(\frac{3}{2})$ located at $(0.50048,-0.296723)$ 
gives a lower bound for the global minimum. This value 
confirms $\hat{F}_{\lambda_r}(a)\geq -(2-\frac{\pi^2}{6})$ first for 
$|\lambda_r|=\frac{1}{4}$ and then, since  $\hat{F}_{\lambda_r}(a)$
decreases in $|\lambda_r|$, for all $0\leq |\lambda_r|\leq \frac{1}{4}$.
We have thus established:
\begin{lemma}
Let $-\frac{1}{6}\leq \lambda\leq 0$ and $f\in
\mathcal{K}_\lambda$. Then
\begin{align}
&|\lambda|\pi a \cot \Big(\frac{|\lambda|\pi}{1-2|\lambda|}\Big)
+ 1 + |\lambda| F(a) 
\leq (Rf)(a) \leq |\lambda|\pi a  \cot (|\lambda|\pi) +1\;,
\qquad 
\text{where} \nonumber \\
&
F(a):=
\frac{4+a }{(1+a)^{\frac{1}{4}}}-4
+\frac{1}{(1+a)^{\frac{1}{4}}}
\Big( \frac{a}{2}
-\frac{4a}{5(1+a)}\;
{}_2F_1\Big(\di{1,\frac{5}{4}}{\frac{9}{4}}
\Big| \frac{1}{1+a}\Big)\Big)\;.
\label{Fa}
\end{align}
\label{Lemma:Fa}
\end{lemma}
We prove:
\begin{lemma}
The function $F(a)$ defined in (\ref{Fa}) has the following
properties:
\begin{enumerate}
\item $F(a)$ is monotonously increasing for $a\geq \frac{1}{2}$.

\item $F(a)$ is convex for $0\leq a \leq \frac{9}{4}$.

\item $F(a)$ is concave for $a \geq \frac{5}{2}$.

\item $|F''(a)|<\frac{1}{10}$ for $\frac{9}{4} \leq a \leq
  \frac{5}{2}$.

\item $F(a) \geq 0$ for $a\geq \frac{4}{5}$.

\item $F(a) \geq -\frac{1}{5}$ for all $a\geq 0$.
\end{enumerate}
\end{lemma}
\begin{proof}
Recall that $F(a)= -\frac{1}{|\lambda_r|}+
\frac{1}{(1+a)^{|\lambda_r|}}\big(\frac{1}{|\lambda_r|}+a
+\hat{F}_{\lambda_r}(a)\big)\big|_{|\lambda_r|=\frac{1}{4}}$.
Differentiation gives with (\ref{F-prime})
\begin{align*}
F'(a)
&=\frac{1}{(1+a)^{\frac{5}{4}}}
\Big( 
\frac{1}{2}+\frac{9}{8}a
-\frac{16}{45(1+a)}\;
{}_2F_1\Big(\di{2,\frac{5}{4}}{\frac{13}{4}}
\Big| \frac{1}{1+a}\Big)
+ \frac{a}{5(1+a)}\;
{}_2F_1\Big(\di{1,\frac{5}{4}}{\frac{9}{4}}
\Big| \frac{1}{1+a}\Big)
\Big)\;.
\end{align*}
This implies the following estimate valid for $a\geq \frac{1}{2}$,
\begin{align*}
F'(a) \geq 
\frac{1}{(1+a)^{\frac{5}{4}}}
\Big( 
\frac{1}{2}+\frac{9}{8}a
-\underbrace{\frac{32}{135}\;
{}_2F_1\Big(\di{2,\frac{5}{4}}{\frac{13}{4}}
\Big| \frac{2}{3}\Big)}_{=0.507407}
+ \frac{a}{5(1+a)}\;
{}_2F_1\Big(\di{1,\frac{5}{4}}{\frac{9}{4}}
\Big| \frac{1}{1+a}\Big)
\Big)\;,
\end{align*}
which shows that $F$ is monotonously increasing for 
all $a\geq \frac{1}{2}$. The second derivative reads with 
(\ref{F-prime})+(\ref{F-primeprime})
\begin{align*}
F''(a)&=
\frac{1}{(1+a)^{\frac{9}{4}}}
\Big( \frac{16-9a}{32}
+\frac{8-a}{9(1+a)}\;
{}_2F_1\Big(\di{2,\frac{5}{4}}{\frac{13}{4}}
\Big| \frac{1}{1+a}\Big)
\nonumber
\\
&+\frac{32}{9\cdot 13a(1+a)}\;
{}_2F_1\Big(\di{2,\frac{5}{4}}{\frac{17}{4}}
\Big| \frac{1}{1+a}\Big)
-
 \frac{5a}{36(1+a)}\;
{}_2F_1\Big(\di{1,\frac{5}{4}}{\frac{13}{4}}
\Big| \frac{1}{1+a}\Big)
\Big)\;.
\end{align*}
Using ${}_2F_1\Big(\di{1,\frac{5}{4}}{\frac{13}{4}}
\Big| \frac{1}{1+a}\Big)\leq 
{}_2F_1\Big(\di{2,\frac{5}{4}}{\frac{13}{4}}
\Big| \frac{1}{1+a}\Big)$ and the lower bound $1$ for the
hypergeometric functions we have the following lower bound for $F''$:
\begin{align*}
F''(a)& \geq 
\frac{1}{(1+a)^{\frac{9}{4}}}
\Big( \frac{16-9a}{32}
+\frac{32-9a}{36(1+a)}
+\frac{32}{117 a(1+a)}\Big)\;.
\end{align*}
This proves that $F(a)$ is convex for all $0\leq a \leq 2.26204$, and
we have 
$F''(a) \geq -\frac{1}{10}$ for all $\frac{9}{4}\leq a \leq
\frac{5}{2}$. 

We derive the converse inequality for $a\geq \frac{9}{4}$ by splitting
the prefactor $\frac{8-a}{9(1+a)}$ at $\frac{9}{4}$. We estimate the
positive hypergeometric functions by its value at $\frac{9}{4}$ and
the negative hypergeometric functions by $1$:
\begin{align*}
a>\tfrac{9}{4}:\quad F''(a)& \leq 
\frac{1}{(1+a)^{\frac{9}{4}}}
\Big( \frac{16-9a}{32}
+\frac{\frac{9}{4}-a}{9(1+a)}
- \frac{5a}{36(1+a)}
\nonumber
\\*
&+\underbrace{\frac{512}{(117)^2}\;
{}_2F_1\Big(\di{2,\frac{5}{4}}{\frac{17}{4}}
\Big| \frac{4}{13}\Big)}_{=0.0458811}
+\underbrace{\frac{23}{117}\;
{}_2F_1\Big(\di{2,\frac{5}{4}}{\frac{13}{4}}
\Big| \frac{4}{13}\Big)}_{= 0.258398}
\Big)\;.
\end{align*}
This proves that $F$ is concave for $a\geq 2.48142$ and the upper
bound $F''(a)\leq 0.08$ for all $\frac{9}{4} \leq a \leq \frac{5}{2}$.

One has $F(1)=0.141693$ and then a good upper bound for $F(t_0)=0$
by the tangent to $F$ at 1, $F(1)+(\tilde{t}_0-1)F'(1)=0$. This shows
$t_0< \frac{4}{5}$. The tangent to $F$ at $\frac{1}{4}$ has positive
slope, the tangent at $\frac{1}{5}$ has negative slope. This means
that the value $F(t_m)$ at the intersection 
of these tangents
$F(\frac{1}{5})+(\tilde{t}_m-\frac{1}{5})F'(\frac{1}{5})=
F(\frac{1}{4})+(\tilde{t}_m-\frac{1}{4})F'(\frac{1}{4})$ gives a lower
bound for $F$. One finds $t_m= 0.223714$ and $F(t_m)= -0.190334$.
\hfill $\square$%
\end{proof}

We have now collected all information to prove:
\begin{lemma}
\begin{align}
F(a) 
\geq 
S(a) := \left\{ 
\begin{array}{c@{\qquad\text{for}\quad}c}
F(\frac{1}{5}) +(a-\frac{1}{5})
F'(\frac{1}{5})  & 0\leq a \leq \frac{1}{2}
\\
F(\frac{3}{2}) +(a-\frac{3}{2})
F'(\frac{3}{2})  & \frac{1}{2} < a < 6
\\
F(6) &  a \geq 6
\end{array}\right.
\label{Sa}
\end{align}
\label{Lemma:Sa}
\end{lemma}
\begin{proof}
The region $0\leq a \leq \frac{1}{2}$ follows from convexity of $F$,
the region $a\geq 6$ because $F$ is monotonously increasing for $a\geq
\frac{1}{2}$. In the intermediate region we have $F(a)\geq S(a)$ 
at least for $\frac{1}{2}\leq a \leq \frac{9}{4}$ because of convexity
of $F$. For $\frac{5}{2}\leq a\leq 6$ we know by concavity that 
\[
F(a)\geq \frac{(a-\tfrac{5}{2})  F(6)+(6-a)F(\tfrac{5}{2})}{6-\tfrac{5}{2}}
\qquad \text{for all } \tfrac{5}{2}\leq a \leq 6\;.
\]
Inserting the numerical values one checks that the secant 
$\frac{(a-\frac{5}{2})  F(6)+(6-a)F(\frac{5}{2})}{6-\frac{5}{2}}
$ lies above the tangent 
$F(\frac{3}{2}) +(a-\frac{3}{2})
F'(\frac{3}{2})$ for $\frac{5}{2}\leq a\leq 6$. There remains the gap 
$\frac{9}{4}\leq a\leq \frac{5}{2}$ where $F$ changes from convex to
concave. Using the bound $|F''(a)|<\frac{1}{10}$ in that region 
we have 
\[
F(a) \geq 
F(\tfrac{9}{4})
+(t-\tfrac{9}{4})F'(\tfrac{9}{4})-\frac{(t-\tfrac{9}{4})^2}{2}
\cdot \frac{1}{10}
\]
for all $\frac{4}{9}\leq a\leq \frac{5}{2}$. The parabola on the rhs 
lies above the tangent 
$F(\frac{3}{2}) +(a-\frac{3}{2})
F'(\frac{3}{2})$. \hspace*{\fill}$\square$%
\end{proof}

Observe that (\ref{Tfa}) implies $Tf(0)=0$ and 
\begin{align}
Tf'(b)&=-\frac{1}{1+b}+
|\lambda| \int_0^{\infty}\!
\frac{dt}{(|\lambda|\pi t)^2 + (b+Rf(t))^2}\;.
\label{Tf-prime}
\end{align}
The inequality of Lemma~\ref{Lemma:Fa} together with the lower bound 
(\ref{Sa}) are now used to derive bounds for $Tf'(b)$.
The inequality $Rf(t)\leq 1+|\lambda|\pi\cot(|\lambda|\pi)$ leads to
\begin{align}
f\in \mathcal{K}_\lambda\quad \Rightarrow \quad Tf'(b) &\geq -\frac{1}{1+b}+ 
\int_0^\infty dt \frac{|\lambda|}{(|\lambda\pi t)^2 
+ (1+b+|\lambda|\pi t\cot(|\lambda|\pi))^2}
\nonumber\\
&= -\frac{1-|\lambda|}{1+b}\;.
\label{Tf-lower}
\end{align}
We thus confirm that $T$ preserves the lower bound of $\mathcal{K}_\lambda$. 
Proving that $T$ preserves the other bound, i.e.\ 
$Tf'(b) +\frac{1-\frac{|\lambda|}{1-2|\lambda}}{1+b} \leq 0$, 
is more difficult. We insert the inequality 
$Rf(t)\geq 1+|\lambda|\pi\cot(|\lambda_r|\pi) 
+ |\lambda|S(a)$ into (\ref{Tfa}) and evaluate the pieces via 
$\displaystyle \int \frac{dt}{(\alpha t)^2+(\beta+\gamma t)^2}=\frac{
\arctan\big(\frac{\alpha t}{\beta + \gamma t} \big)}{ \alpha\beta}$.
This gives for any $f\in \mathcal{K}_\lambda$ and with partial use of 
$|\lambda_r|:=\frac{|\lambda|}{1-2|\lambda|}$:
\begin{align}
&Tf'(b) +\frac{1-\frac{|\lambda|}{1-2|\lambda}}{1+b}
\nonumber
\\
&\leq \int_0^\infty dt \frac{|\lambda|}{(|\lambda\pi t)^2 
+ (1+b+|\lambda|S(t)+|\lambda|\pi t\cot(|\lambda_r|\pi))^2}
-\frac{|\lambda_r|}{1+b}
\nonumber
\\
&=
\int_0^{\frac{1}{2}} dt \frac{|\lambda|}{(|\lambda|\pi t)^2 
+ \big((b{+}1+|\lambda|F(\frac{1}{5}) {-} \frac{|\lambda|}{5} 
F'(\frac{1}{5})) +
(|\lambda|t F'(\frac{1}{5}) 
+|\lambda|\pi t\cot(|\lambda_r|\pi)\big)^2}
\nonumber
\\
& + \int_{\frac{1}{2}}^6 dt \frac{|\lambda|}{(|\lambda|\pi t)^2 
+ \big((b{+}1+|\lambda|F(\frac{3}{2}) {-} \frac{3|\lambda|}{2} 
F'(\frac{3}{2})) +
(|\lambda|t F'(\frac{3}{2}) 
+|\lambda|\pi t\cot(|\lambda_r|\pi)\big)^2}
\nonumber
\\
&+ \int_6^\infty dt \frac{|\lambda|}{(|\lambda|\pi t)^2 
+ \big((b+1+|\lambda|F(6) )
+|\lambda|\pi t\cot(|\lambda_r|\pi)\big)^2}
- \frac{|\lambda_r|\pi}{\pi(1+b)}
\nonumber
\\
&= \frac{\arctan \Big( \frac{\frac{1}{2}|\lambda|\pi }{
b+1+|\lambda|F(\frac{1}{5}) + \frac{3|\lambda|}{10} 
F'(\frac{1}{5}) 
+\frac{1}{2}|\lambda|\pi \cot(|\lambda_r|\pi)}\Big)}{\pi 
(b+1+|\lambda|F(\frac{1}{5}) - \frac{|\lambda|}{5} 
F'(\frac{1}{5}))}
\nonumber
\\
&
- \frac{\arctan \Big( \frac{\frac{1}{2}|\lambda|\pi }{
b+1+|\lambda|F(\frac{3}{2}) - |\lambda|
F'(\frac{3}{2}) +
\frac{1}{2}|\lambda|\pi \cot(|\lambda_r|\pi)}\Big)}{\pi 
(b+1+|\lambda|F(\frac{3}{2}) - \frac{3|\lambda|}{2} 
F'(\frac{3}{2}))}
\nonumber
\\
&+ \frac{\arctan \Big( \frac{6|\lambda|\pi }{
b+1+|\lambda|F(\frac{3}{2}) + \frac{9|\lambda|}{2} 
F'(\frac{3}{2}) 
+6|\lambda|\pi \cot(|\lambda_r|\pi)}\Big)}{\pi 
(b+1+|\lambda|F(\frac{3}{2}) - \frac{3|\lambda|}{2} 
F'(\frac{3}{2}))}
- \frac{\arctan \Big( \frac{6|\lambda|\pi }{
b+1+|\lambda|F(6) 
+6|\lambda\pi \cot(|\lambda_r|\pi)}\Big)}{\pi 
(b+1+|\lambda|F(6) )}
\nonumber
\\
&
+ \frac{|\lambda_r|\pi}{\pi 
(b+1+|\lambda|F(6) )}
- \frac{|\lambda_r|\pi}{\pi (b+1)}\;.
\end{align}
For $0\leq |\lambda|\leq \frac{1}{6}$ we have $\cot
(|\lambda_r|\pi)\geq 1$. We are therefore within the convergence
domain of the $\arctan$ series, and Leibniz' criterion gives upper and
lower bounds:
\[
0\leq x\leq 1\quad \Rightarrow\quad 
x-\frac{x^3}{3} \leq  \arctan x \leq 
x-\frac{x^3}{3} +\frac{x^5}{5}\;.
\]
For the sake of transparence we abbreviate
\begin{align*}
\beta &:=\tfrac{b+1}{|\lambda|\pi}\;, &
\gamma &:= \cot(|\lambda_r|\pi)\;, &
\\
\delta_1 &:= \tfrac{1}{\pi}F(\tfrac{1}{5}) + \tfrac{3}{10\pi} F'(\tfrac{1}{5})\;, &
\delta_2 &:= \tfrac{1}{\pi}F(\tfrac{1}{5}) - \tfrac{1}{5\pi} F'(\tfrac{1}{5})\;, &
\delta_3 &:= \tfrac{1}{\pi}F(\tfrac{3}{2}) - \tfrac{1}{\pi}F'(\tfrac{3}{2})\;,
\\
\delta_4&:= \tfrac{1}{\pi}F(\tfrac{3}{2}) - \tfrac{3}{2\pi} F'(\tfrac{3}{2})\;,&
\delta_5&:= \tfrac{1}{\pi}F(\tfrac{3}{2}) + \tfrac{9}{2\pi} F'(\tfrac{3}{2})\;,&
\delta_6&:= \tfrac{1}{\pi}F(6)\;.
\end{align*}
Then 
\begin{align}
&|\lambda|\pi^2 \Big(Tf'(b) +\frac{1-\frac{|\lambda|}{1-2|\lambda}}{1+b}\Big)
\nonumber
\\
&\leq\frac{1}{(2\beta{+}2\delta_1{+}\gamma)(\beta{+}\delta_2)}
-\frac{1}{3(2\beta{+}2\delta_1{+}\gamma)^3 (\beta{+}\delta_2)}
+\frac{1}{5(2\beta{+}2\delta_1{+}\gamma)^5 (\beta{+}\delta_2)}
\nonumber
\\
&-\frac{1}{(2\beta{+}2\delta_3{+}\gamma)(\beta{+}\delta_4)}
+\frac{1}{3(2\beta{+}2\delta_3{+}\gamma)^3 (\beta{+}\delta_4)}
\nonumber
\\
&+\frac{1}{(\frac{1}{6}\beta{+}\frac{1}{6}\delta_5{+}\gamma)
(\beta{+}\delta_4)}
-\frac{1}{3(\frac{1}{6}\beta{+}\frac{1}{6}\delta_5{+}\gamma)^3
(\beta{+}\delta_4)}
+\frac{1}{5(\frac{1}{6}\beta{+}\frac{1}{6}\delta_5{+}\gamma)^5
(\beta{+}\delta_4)}
\nonumber
\\
&-\frac{1}{(\frac{1}{6}\beta{+}\frac{1}{6}\delta_6{+}\gamma)
(\beta{+}\delta_6)}
{+}\frac{1}{3(\frac{1}{6}\beta{+}\frac{1}{6}\delta_6{+}\gamma)^3
(\beta{+}\delta_6)}
{+}\frac{|\lambda_r|\pi}{(\beta+\delta_6)}
-\frac{|\lambda_r|\pi}{\beta}
\nonumber
\\
&=: \frac{\sum_{k=0}^{18} c_k
\big(\beta-\frac{1}{|\lambda\pi|}\big)^k }{
(\beta{+}\delta_1{+}\frac{1}{2}\gamma)^5 
(\beta{+}\delta_3{+}\frac{1}{2}\gamma)^3 
(\beta{+}\delta_5{+}6\gamma)^5 
(\beta{+}\delta_6{+}6\gamma)^3 
(\beta{+}\delta_2)
(\beta{+}\delta_4)(\beta{+}\delta_6)\beta}.
\end{align}
In the last line, the coefficients $c_k$ are 
polynomials in $\gamma,|\lambda_r\pi|$ and 
$\frac{1}{|\lambda|\pi}$. One finds with
$|\lambda_r|\cot(|\lambda_r|\pi) \geq \frac{1}{4}$ and 
$|\lambda_r|-|\lambda|\geq 0$ for all $0\leq
|\lambda|\leq \frac{1}{6}$:
\begin{align}
c_{18}&= -|\lambda_r|\pi \delta_6 = -3.53 |\lambda_r|\;,
\nonumber
\\
c_{17}&= - 29.01 |\lambda_r| - 183.74 |\lambda_r| \cot(|\lambda_r|\pi)
- \tfrac{20.25 |\lambda_r|-7.75|\lambda|}{|\lambda|}\;,
\nonumber
\\
c_{16}&=   - 101.11 |\lambda_r| - 1318.89|\lambda_r| \cot
(|\lambda_r|\pi)- 
 4264.94 |\lambda_r| \cot^2(|\lambda_r|\pi) 
\nonumber
\\
& - \tfrac{54.78 |\lambda_r|-41.92 |\lambda|}{|\lambda|^2} 
- \tfrac{156.99 |\lambda_r|-56.11 |\lambda|}{|\lambda|}
- 
\tfrac{994.28 |\lambda_r|- 355.99|\lambda|}{|\lambda|} \cot (|\lambda_r|\pi)
\;,
\nonumber
\\
c_{15}&= 
-426.99 \tfrac{|\lambda_r| \cot (|\lambda_r|\pi)-\frac{1}{4}}{
|\lambda|^2} 
- 191.62 |\lambda_r| - 3914.61 |\lambda_r| \cot (|\lambda_r|\pi) 
\nonumber
\\
&- 
 26296.4 |\lambda_r| \cot^2 (|\lambda_r|\pi) 
- 
|\lambda_r| \cot^3 (|\lambda_r|\pi) 
- \tfrac{92.992 |\lambda_r|}{|\lambda|^3} - 
\tfrac{399.76 |\lambda_r|- 285.75 |\lambda|}{|\lambda|^2}
\nonumber
\\
&- \tfrac{2104.91 |\lambda_r| - 
 1813.03 |\lambda|}{|\lambda|^2} \cot (|\lambda_r|\pi)
- \tfrac{514.93
   |\lambda_r|-74.85|\lambda|}{|\lambda|} 
\nonumber
\\
&- 
\tfrac{6717.04 |\lambda_r|- 2214.36 |\lambda|}{|\lambda|} 
\cot (|\lambda_r|\pi)
- \tfrac{21721.1 |\lambda_r|- 7195.88
   |\lambda|}{|\lambda|} \cot^2 (|\lambda_r|\pi)\;,
\nonumber
\\
c_{14}&= 
-5405 \big(|\lambda_r|\cot(|\lambda_r|\pi)-\tfrac{1}{4}\big)
- 2729
\tfrac{|\lambda_r|\cot(|\lambda_r|\pi)-\frac{1}{4}}{|\lambda_r|^2}
- 679.6 \tfrac{|\lambda_r|\cot(|\lambda_r|\pi)-\frac{1}{4}}{|\lambda_r|^3}
\nonumber
\\
&- 17313 \tfrac{|\lambda_r|\cot (|\lambda_r|\pi)-\frac{1}{4}}{
|\lambda|^2}\cot (|\lambda_r|\pi)
- 207.1 |\lambda_r| - 651.1 |\lambda_r| \cot (|\lambda_r|\pi)
\nonumber
\\
&
-  64754 |\lambda_r| \cot^2 (|\lambda_r|\pi)
- 301494 |\lambda_r| \cot^3 (|\lambda_r|\pi)
- 517659 |\lambda_r|\cot^4 (|\lambda_r|\pi) 
\nonumber
\\
&- \tfrac{111 |\lambda_r|}{|\lambda|^4} - \tfrac{636.239
  |\lambda_r|}{|\lambda|^3} 
-\tfrac{3350 |\lambda_r|}{|\lambda|^3} \cot (|\lambda_r|\pi) 
-\tfrac{1229 |\lambda_r|-
  357.4|\lambda|}{|\lambda|^2} 
- \tfrac{914.9 |\lambda_r|}{|\lambda|}
\nonumber
\\
&- \tfrac{13307 |\lambda_r| {-} 10573 |\lambda|}{|\lambda|^2} 
\cot (|\lambda_r|\pi)
- \tfrac{34542 |\lambda_r|{-} 34358 |\lambda|}{|\lambda|^2}  
\cot^2 (|\lambda_r|\pi)
- \tfrac{18691 |\lambda_r|{-} 2338|\lambda|}{|\lambda|} 
\cot (|\lambda_r|\pi)
\nonumber
\\
&
- \tfrac{125556 |\lambda_r|- 37587 |\lambda|}{|\lambda|} 
\cot^2 (|\lambda_r|\pi)
- \tfrac{277886|\lambda_r|-
84001 |\lambda|}{|\lambda|} \cot^3 (|\lambda_r|\pi)\;.
\end{align}
All contributions are manifestly negative. That negativity continues to 
all $c_k$, but the expressions become of exceeding length. It does not
make much sense to display these formulae. Instead we give in
Figure~\ref{fig:ck} a graphical discription of the coefficients $c_k$.
\begin{figure}[ht]
\begin{picture}(120,126)
  \put(0,60){\includegraphics[width=3.6cm,bb= 0 0 209 312]{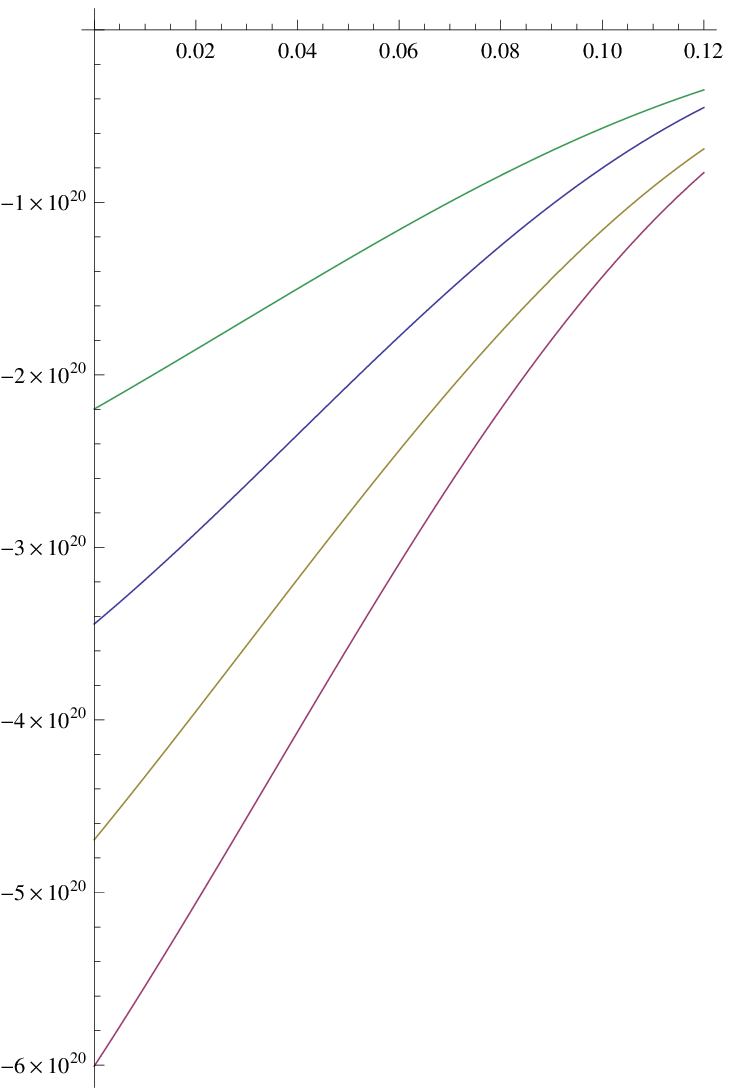}}
  \put(12,92){\mbox{\small$c_{0}'$}}
  \put(18,79){\mbox{\small$c_{1}'$}}
  \put(8,84){\mbox{\small$c_{2}'$}}
  \put(8,107){\mbox{\small$c_{3}'$}}
  \put(42,60){\includegraphics[width=2.42cm,bb= 0 0 141 312]{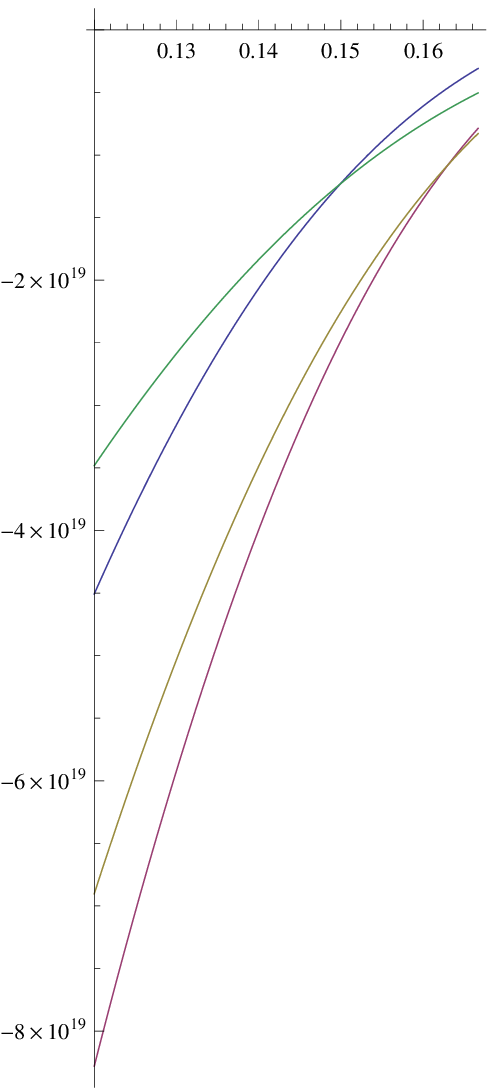}}
  \put(52,98){\mbox{\small$c_{0}'$}}
  \put(54,79){\mbox{\small$c_{1}'$}}
  \put(50,90){\mbox{\small$c_{2}'$}}
  \put(49,107){\mbox{\small$c_{3}'$}}
  \put(75,60){\includegraphics[width=5.85cm,bb= 0 0 260 237]{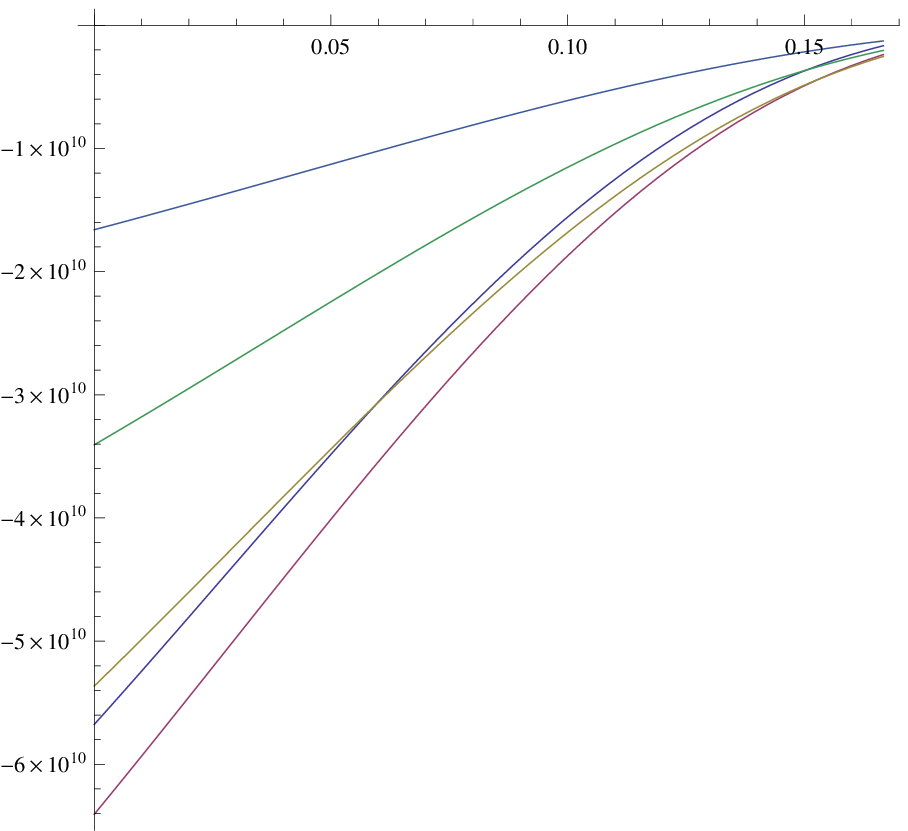}}
  \put(90,76){\mbox{\small$c_{4}''$}}
  \put(100,80){\mbox{\small$c_{5}''$}}
  \put(89,83){\mbox{\small$c_{6}''$}}
  \put(100,105){\mbox{\small$c_{7}''$}}
  \put(100,116){\mbox{\small$c_{8}''$}}
  \put(0,0){\includegraphics[width=5.85cm,bb= 0 0 260 196]{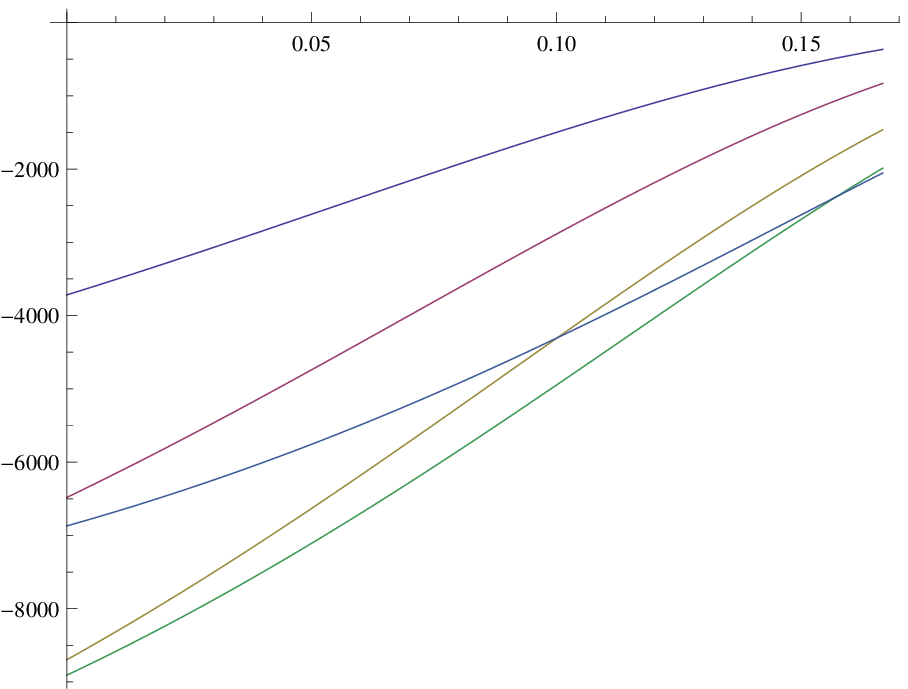}}
  \put(27,42){\mbox{$c_{9}'''$}}
  \put(21,28){\mbox{$c_{10}'''$}}
  \put(9,8){\mbox{$c_{11}'''$}}
  \put(28,11){\mbox{$c_{12}'''$}}
  \put(36,28){\mbox{$c_{13}'''$}}
  \put(75,0){\includegraphics[width=5.85cm,bb= 0 0 260 199]{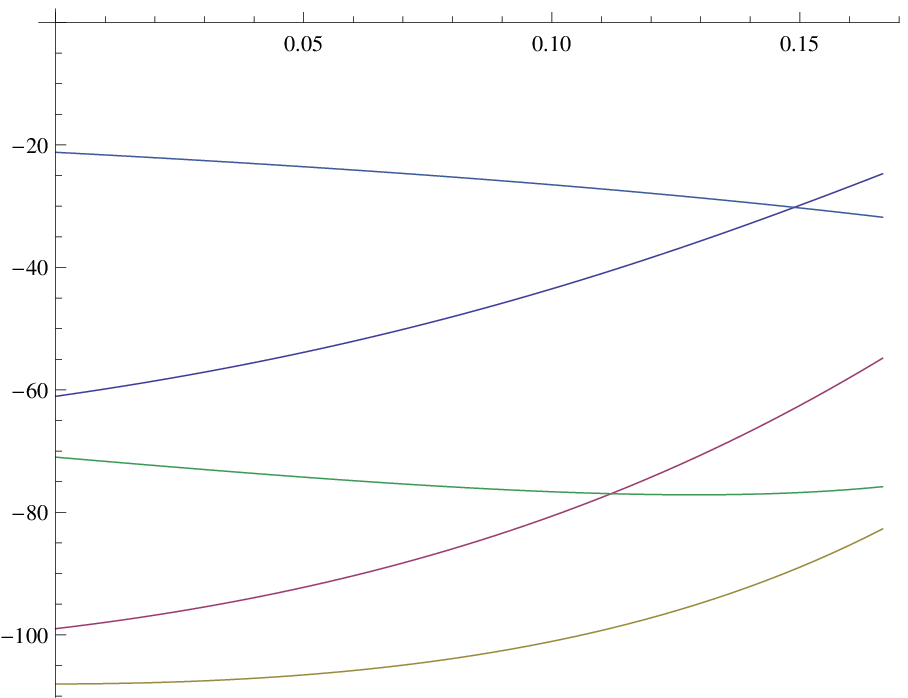}}
  \put(115,45){\mbox{$c_{18}''''$}}
  \put(110,19.5){\mbox{$c_{17}''''$}}
  \put(104,32.5){\mbox{$c_{14}''''$}}
  \put(90,10){\mbox{$c_{15}''''$}}
  \put(140,8){\mbox{$c_{16}''''$}}
\end{picture}
\caption{Plot of the rescaled coefficients 
$c_k'(\lambda):= \big(20\lambda\big)^{17-k} c_k$,
$c_k''(\lambda):= \big(4\lambda\big)^{17-k} c_k$,
$c_k'''(\lambda):= \big(\frac{\lambda}{6}\big)^{17-k} c_k$
and $c_k''''(\lambda):= \big(\frac{2\lambda}{3}\big)^{17-k} c_k$. All
of them are manifestly negative.
\label{fig:ck}}
\end{figure}
We confirm that all of them are negative for any $0\leq |\lambda|\leq
\frac{1}{6}$, thus proving 
\[
Tf'(b) \leq -\frac{1-\frac{|\lambda|}{1-2|\lambda}}{1+b}\;
\qquad\text{for all } b\geq 0 \text{ and any }
-\frac{1}{6}\leq \lambda\leq 0\;.
\]
This finishes the proof that $T$ maps $\mathcal{K}_\lambda$ into itself.

\section{$T$ is uniformly continuous 
on $\mathcal{K}_\lambda$, but not 
contractive}

Take $f,g\in LB$ with $\|f-g\|_{LB}:=\delta$. This means
$
-\frac{\delta}{1+x} \leq  f'(x)-g'(x) \leq
\frac{\delta}{1+x}$ for all $x\in \mathbb{R}_+$. 
Integration from $a$ to $x>a$ yields 
\[
-\delta \log \frac{1+x}{1+a} \leq  f(x)-g(x)-f(a)+g(a) \leq  
\delta \log \frac{1+x}{1+a}
\]
or $\displaystyle \Big(\frac{1+a}{1+x} \Big)^{\delta}
\leq  \frac{e^{f(x)}}{e^{f(a)}}
\frac{e^{g(a)}}{e^{g(x)}} \leq \Big(\frac{1+x}{1+a} \Big)^\delta$. Together 
with (\ref{efg-5}) we deduce the following inequalities valid for $x>a$
\begin{align*}
\max\Big\{ \Big(\frac{1+x}{1+a} \Big)^{|\lambda|-1},
\Big(\frac{1+x}{1+a} \Big)^{-\delta}\frac{e^{g(x)}}{e^{g(a)}} \Big\}
&\leq  \frac{e^{f(x)}}{e^{f(a)}}
\\
&\leq \min\Big\{ \Big(\frac{1+x}{1+a} \Big)^{\frac{|\lambda|}{2-|\lambda|}-1},
\Big(\frac{1+x}{1+a} \Big)^\delta 
\frac{e^{g(x)}}{e^{g(a)}} \Big\}\;.
\end{align*}
We subtract $\frac{e^{g(x)}}{e^{g(a)}}=: \big(\frac{1+x}{1+a} \big)^{\mu-1}$ 
with $|\lambda| \leq \mu \leq \frac{|\lambda|}{1-2|\lambda|}$. 
A careful discussion of $\mu$ versus $|\lambda|+\delta$ shows that 
for $x>a$ one has
\begin{align}
{-} \Big(  \Big(\frac{1 {+}x}{1{+}a} \Big)^{\frac{|\lambda|}{1-2|\lambda|}-1}
{-} \Big(\frac{1{+}x}{1{+}a} \Big)^{\frac{|\lambda|}{1-2|\lambda|}-\delta-1}\Big)
&\leq \frac{e^{f(x)}}{e^{f(a)}}-
\frac{e^{g(x)}}{e^{g(a)}} 
\nonumber
\\
&\leq 
\Big(\frac{1{+}x}{1{+}a} \Big)^{\frac{|\lambda|}{1-2|\lambda|}-1}
{-} \Big(\frac{1{+}x}{1{+}a} \Big)^{\frac{|\lambda|}{1-2|\lambda|}-\delta-1}.
\label{efg-1m}
\end{align}
Conversely, for $x<a$ we start from 
$\displaystyle \Big(\frac{1+a}{1+x} \Big)^{-\delta}
\leq \frac{e^{f(x)}}{e^{f(a)}}
\frac{e^{g(a)}}{e^{g(x)}} \leq \Big(\frac{1+a}{1+x} \Big)^\delta$, 
which together with (\ref{efg-5}) leads to
\begin{align*}
\max\Big\{ \Big(\frac{1+a}{1+x} \Big)^{1-\frac{|\lambda|}{1-2|\lambda|}},
\Big(\frac{1+a}{1+x} \Big)^{-\delta}\frac{e^{g(x)}}{e^{g(a)}} \Big\}
&\leq  \frac{e^{f(x)}}{e^{f(a)}}
\\
&\leq \min\Big\{ \Big(\frac{1+a}{1+x} \Big)^{1-|\lambda|},
\Big(\frac{1+a}{1+x} \Big)^\delta 
\frac{e^{g(x)}}{e^{g(a)}} \Big\}\;.
\end{align*}
We subtract $\frac{e^{g(x)}}{e^{g(a)}}=: \big(\frac{1+a}{1+x}
\big)^{1-\mu}$ with $|\lambda| \leq \mu \leq
\frac{|\lambda|}{1-2|\lambda|}$.  A careful discussion of $\mu$ versus
$\frac{|\lambda|}{1-2|\lambda|}-\delta$ shows that for $x<a$ one has
\begin{align}
- \Big(  \Big(\frac{1+a}{1+x} \Big)^{1-|\lambda|}
- \Big(\frac{1+a}{1+x} \Big)^{1-|\lambda|-\delta}\Big)
&\leq \frac{e^{f(x)}}{e^{f(a)}}-
\frac{e^{g(x)}}{e^{g(a)}} 
\nonumber
\\ 
&  \leq 
\Big(\frac{1+a}{1+x} \Big)^{1-|\lambda|}
- \Big(\frac{1+a}{1+x} \Big)^{1-|\lambda|-\delta}\;.
\label{efg-1n}
\end{align}

With these preparations we can prove that 
$\frac{\mathcal{H}^{\!\infty}_a[e^{f(\bullet)}]}{e^{f(a)}}$ varies slowly with $f$:
\begin{lemma}
  For any $f,g\in \mathcal{K}_\lambda$ with $\|f-g\|_{LB} =\delta$,
  hence $\delta \leq \frac{2|\lambda|^2}{1-2|\lambda|}$, one has for
  $0\leq |\lambda|<\frac{1}{3}$ the bound
\begin{align}
\Big|\frac{\mathcal{H}^{\!\infty}_a[e^{f(\bullet)}]}{e^{f(a)}}
- \frac{\mathcal{H}^{\!\infty}_a[e^{g(\bullet)}]}{e^{g(a)}}\Big|
&< \delta \cdot 
\Big( \zeta_\lambda 
+\frac{1}{|\lambda|\pi} 
\cdot 
\frac{(1+a)^{|\lambda|}-1-|\lambda|\log(1+a)}{
|\lambda|(1+a)^{|\lambda|}} \Big)\;,
\\
\zeta_\lambda&:= \frac{1}{\pi} \sum_{k=1}^\infty \Big(
\frac{1}{(k+|\lambda|)^2}+\frac{1}{(k-\frac{|\lambda|}{1-2|\lambda|})^2}\Big)\;.
\nonumber
\end{align}
\label{Lemma-zeta}
\end{lemma}
\begin{proof} We take the Hilbert transform of (\ref{efg-1m}) and
  (\ref{efg-1n}).  The principal value limit can be weakened to
  improper Riemann integrals:
\begin{align}
\Big|\frac{\mathcal{H}^{\!\infty}_a[e^{f(\bullet)}]}{e^{f(a)}}
- \frac{\mathcal{H}^{\!\infty}_a[e^{g(\bullet)}]}{e^{g(a)}}\Big|
& \leq
\int_0^a \frac{dx}{\pi}\;
\frac{\Big(\dfrac{1+x}{1+a}\Big)^{|\lambda|-1}
-\Big(\dfrac{1+x}{1+a}\Big)^{|\lambda|+\delta-1}}{(1+a)-(1+x)} 
\nonumber
\\
&+
\int_a^\infty \frac{dx}{\pi}\;
\frac{\Big(\dfrac{1+x}{1+a}\Big)^{\frac{|\lambda|}{1-2|\lambda|}-1}
-\Big(\dfrac{1+x}{1+a}\Big)^{\frac{|\lambda|}{1-2|\lambda|}-\delta-1}
}{(1+x)-(1+a)} 
\nonumber
\\
& 
=\underbrace{\lim_{\epsilon\to 0}\Big( 
\int_{\frac{1}{1+a}}^{1-\epsilon} \frac{dt}{\pi}\;
\frac{t^{|\lambda|-1}}{1-t} 
-\int_{\frac{1}{1+a}}^{1-\epsilon} \frac{dt}{\pi}\;
\frac{t^{|\lambda|+\delta-1}}{1-t} \Big)}_{I_1}
\nonumber
\\
&
+\underbrace{\int_1^\infty \frac{dt}{\pi}\;
\frac{t^{\frac{|\lambda|}{1-2|\lambda|}-1}
-t^{{\frac{|\lambda|}{1-2|\lambda|}-\delta-1}}}{t-1} }_{I_2}
\label{cont-HTinfty}\;. 
\end{align}
The second integral $I_2$ known from \cite[\S 3.231.6+\S 3.231.5]{Gradsteyn:1994??}:
\begin{subequations}
\begin{align}
I_2
&=\cot \big(\tfrac{|\lambda|}{1-2|\lambda|}\pi{-}\delta\pi\big)-
\cot \big(\tfrac{|\lambda|}{1-2|\lambda|}\pi\big)
-\frac{1}{\pi} \Big( \psi(\tfrac{|\lambda|}{1-2|\lambda|})
-\psi(\tfrac{|\lambda|}{1-2|\lambda|}{-}\delta))\Big)
\nonumber
\\
&= \frac{\psi(1-\tfrac{|\lambda|}{1-2|\lambda|}{+}\delta)
- \psi(1-\tfrac{|\lambda|}{1-2|\lambda|})}{\pi}
\nonumber
\\
&=\frac{1}{\pi}\sum_{k=0}^\infty\Big( \frac{1}{k+1-\tfrac{|\lambda|}{1-2|\lambda|}}
-\frac{1}{k+1-\tfrac{|\lambda|}{1-2|\lambda|}+\delta}
\Big)
\nonumber
\\
&\leq \frac{\delta}{\pi}\sum_{k=1}^\infty \frac{1}{
\big(k-\tfrac{|\lambda|}{1-2|\lambda|}\big)^2}
\;.
\label{zeta-plus}
\end{align}
We have used the power series expansion \cite[\S
8.363.3]{Gradsteyn:1994??} for the difference of digamma functions.
The result is uniformly bounded for $|\lambda|<\frac{1}{3}$.

The first integral $I_1$ is evaluated with (\ref{HyperG-}) to
\begin{align}
I_1&= 
\frac{1}{\pi} \lim_{\epsilon\to 0}\Big\{ 
\frac{(1{-}\epsilon)^{|\lambda|}}{|\lambda|}
{}_2F_1\Big(\di{1,|\lambda|}{1{+}|\lambda|}\Big| 
1{-}\epsilon\Big)
- 
 \frac{(1{-}\epsilon)^{|\lambda|+\delta}}{|\lambda|+\delta}
{}_2F_1\Big(\di{1,|\lambda|{+}\delta}{1{+}|\lambda|{+}\delta}\Big| 
1{-}\epsilon\Big)
\nonumber
\\
& - \frac{(\frac{1}{1+a})^{|\lambda|}}{|\lambda|}
{}_2F_1\Big(\di{1,|\lambda|}{1{+}|\lambda|}\Big| 
\frac{1}{1+a}\Big)
+\frac{(\frac{1}{1+a})^{|\lambda|+\delta}}{|\lambda|+\delta} 
{}_2F_1\Big(\di{1,|\lambda|{+}\delta}{1{+}|\lambda|{+}\delta}\Big| 
\frac{1}{1+a}\Big)\Big\}
\nonumber
\\
&= \frac{1-(\frac{1}{1+a})^{|\lambda|}}{|\lambda|\pi}
-\frac{1-(\frac{1}{1+a})^{|\lambda|+\delta}}{(|\lambda|+\delta)\pi}
\label{zeta-0}
\\
&+\frac{1}{\pi} \sum_{k=1}^\infty \Big\{
\frac{1}{k+|\lambda|} \Big(1-\frac{1}{(1+a)^{k+|\lambda|}}\Big)
-\frac{1}{k+|\lambda|{+}\delta} \Big(1-\frac{1}{(1+a)^{k+|\lambda|+\delta}}\Big)
\Big\}\;.
\label{zeta-minus}
\end{align}
\end{subequations}
Here we have expanded 
the hypergeometric functions into a 
power series and rearranged them to differences which admit the limit 
$\epsilon\to 0$.
The line (\ref{zeta-minus}) is monotonous in $a$ and thus can be estimated 
by its limit $a\to \infty$.
The same argument gives a possible uniform estimate of (\ref{zeta-0}).
\begin{align}
(\textup{\ref{zeta-minus}})\leq  \frac{\delta}{\pi}
\sum_{k=1}^\infty \frac{1}{\big(k+|\lambda|\big)^2}\;,\qquad\qquad
(\textup{\ref{zeta-0}})\leq  \frac{\delta}{\pi} \frac{1}{|\lambda|^2}\;.
\tag{\ref{zeta-minus}'+\ref{zeta-0}'}
\end{align}

The last estimate is enough for continuity, but not for contractivity.
We write (\ref{zeta-0}) as a double integral:
\begin{align}
&\frac{1-(\frac{1}{1+a})^{|\lambda|}}{|\lambda|\pi}
-\frac{1-(\frac{1}{1+a})^{|\lambda|+\delta}}{(|\lambda|+\delta)\pi}
=\frac{1}{\pi} \int_{\frac{1}{1+a}}^1 dt\; \big(t^{|\lambda|-1}
-t^{|\lambda|+\delta-1}\big)
\nonumber
\\
&=-\frac{1}{\pi} \int_{\frac{1}{1+a}}^1 dt\; \int_0^\delta d\xi \;\frac{d}{d\xi} 
t^{|\lambda|+\xi-1}
=\frac{1}{\pi} \int_{\frac{1}{1+a}}^1 dt\; \int_0^\delta d\xi \;(-\log t) 
t^{|\lambda|+\xi-1}
\nonumber
\\
&\leq 
\frac{\delta}{\pi} \int_{\frac{1}{1+a}}^1 dt\; (-\log t) 
t^{|\lambda|-1}
=\frac{\delta}{|\lambda|^2\pi} 
\Big(\frac{(1+a)^{|\lambda|}-1-|\lambda|\log(1+a))}{(1+a)^{|\lambda|}} \Big)
\;.
\end{align}
This gives together with 
(\ref{zeta-plus}) and the estimate (\ref{zeta-minus}') the claimed result.
\hfill $\square$%
\end{proof}

Putting $x=0$ in (\ref{efg-1n}) leads to
\begin{align}
\Big|\frac{1}{e^{f(a)}}-\frac{1}{e^{g(a)}}\Big|
& \leq 
(1+a)^{1-|\lambda|}-(1+a)^{1-|\lambda|-\delta}
=-\int_0^\delta d\xi \frac{d}{d\xi} 
(1+t)^{1-|\lambda|-\xi}
\nonumber
\\
& 
=\int_0^\delta d\xi \;
(1+t)^{1-|\lambda|-\xi}
\log (1+t)
\leq \delta 
(1+t)^{1-|\lambda|}
\log (1+t)
\;.
\label{efg-4}
\end{align}

Together with Lemma~\ref{Lemma-zeta} we have thus proved for the map 
$R$ defined in (\ref{Rf}):
\begin{proposition}
\label{Prop:Rfg}
Let $0\leq |\lambda|<\frac{1}{3}$.
For any $f,g\in \mathcal{K}_\lambda$ with 
$\|f-g\|_{LB} =\delta$ one has the pointwise bound
\begin{subequations}
\label{DeltaR}
\begin{align}
\Big|(Rf)(t)-(Rg)(t)\Big|
&\leq  (\Delta R)^{(1)}(t)+ 
(\Delta R)^{(2)}(t)+ 
(\Delta R)^{(3)}(t)\;,
\nonumber
\\
(\Delta R)^{(1)}(t)
&:=  \delta \cdot (1+t)^{1-|\lambda|} \log (1+t)\;,
\\
(\Delta R)^{(2)}(t)
&:=  \delta \cdot |\lambda|\pi t \zeta_\lambda\;,
\\
(\Delta R)^{(3)}(t)
&:=  \delta \cdot t \cdot
\frac{(1+t)^{|\lambda|} -1-|\lambda| \log (1+t) }{
|\lambda|(1+t)^{|\lambda|} }\;.
\end{align}
\end{subequations}
\end{proposition}

\begin{proposition}
\label{prop:continuity}
The map $T:\mathcal{K}_\lambda\to \mathcal{K}_\lambda$ is
norm-continuous. More precisely, for $-\frac{1}{6}\leq \lambda \leq 0$
one has
\begin{align}
\|Tf-Tg\|_{LB} \leq \|f-g\|_{LB} \cdot 
\frac{\sin^2(\frac{|\lambda|\pi}{1-2|\lambda|})}{(|\lambda|\pi)^2}
\frac{(1-\frac{|\lambda|}{5})^{-1}}{\cos(\frac{|\lambda|\pi}{1-2|\lambda|})}
\Big(1+ \frac{1+|\lambda|}{e} + |\lambda|^2 \pi \zeta_\lambda\Big)
\;.
\end{align}
The rhs ranges from 
$1.36788 \|f-g\|_{LB} $ for $|\lambda|=0$ to 
$4.09942 \|f-g\|_{LB} $ for  $|\lambda|=\frac{1}{6}$.
\end{proposition}
\begin{proof}
The definition (\ref{Tf-prime}) gives for $f,g\in \mathcal{K}_\lambda$
\begin{align}
&\|Tf-Tg\|_{LB} 
\nonumber
\\
&= 
\sup_{a\geq 0} 
|\lambda| \int_0^{\infty}\!dt \;
\frac{(1+a)\big|Rg(t)-Rf(t)\big|(2a+Rg(t)+Rf(t))}{\big((|\lambda|\pi t)^2 + (a+Rf(t))^2\big)
\big((|\lambda|\pi t)^2 + (a+Rg(t))^2\big)}
\nonumber
\\
& \leq \sum_{\tau=1}^3 
\sup_{a\geq 0} 
2|\lambda| \int_0^{\infty}\!dt \;
\frac{(1+a) (\Delta R)^{(\tau)}(t)}{
\big((|\lambda|\pi t)^2 
+ \big(a+1+|\lambda|\pi t\cot(\frac{|\lambda|\pi}{1-2|\lambda|}) 
+ |\lambda|F(t))\big)^2 \big)^{\frac{3}{2}}}\;,
\end{align}
where we have inserted the lower bound
$Rf(t)\geq 
1+|\lambda|\pi t\cot(\frac{|\lambda|\pi}{1-2|\lambda|}) 
+ |\lambda|F(t)$ derived in 
Lemma~\ref{Lemma:Fa}.
We write this as corresponding decomposition 
$\|Tf-Tg\|_{LB} \leq \sum_{\tau=1}^3 
\|Tf-Tg\|_{LB}^{(\tau)}$.

We start with the easiest contribution $\tau=2$ where we substitute
$u=|\lambda|\pi t$:
\begin{align*}
\|Tf-Tg\|_{LB}^{(2)} 
&:= 
\sup_{a\geq 0} 
\frac{2\delta}{\pi} \int_0^{\infty}\!du \;
\frac{(1+a) \zeta_\lambda u }{
\big( u^2 + \big(a+1+|\lambda|F(\frac{u}{|\lambda|\pi})
+ u \cot(\frac{|\lambda|\pi}{1-2|\lambda|}) 
\big)^2 \big)^{\frac{3}{2}}}\;.
\end{align*}
There is no doubt that $F(t)$ is of positive mean also for this
integral (the small-$u$-region is suppressed) so that it is safe to
put $F(\,.\,)\mapsto 0$. We postpone this proof and temporarily work
with the conservative estimate
$1+|\lambda|F(\frac{u}{|\lambda|\pi})\geq
h_\lambda:=1-\frac{|\lambda|}{5}$. This reduces the problem to a
standard integral \cite[\S 3.252.7]{Gradsteyn:1994??}:
\begin{align}
&\|Tf-Tg\|_{LB}^{(2)} 
\nonumber
\\
&=
\sup_{a\geq 0} 
\int_0^{\infty}\!du \;
\frac{2\delta \zeta_\lambda (1+a) 
\sin^3(\frac{|\lambda|\pi}{1-2|\lambda|}) \cdot u}{\pi 
\big( u^2 + 2u 
\sin(\frac{|\lambda|\pi}{1-2|\lambda|}) 
\cos(\frac{|\lambda|\pi}{1-2|\lambda|}) 
(a+h_\lambda)
+\big(\sin(\frac{|\lambda|\pi}{1-2|\lambda|}) (a+h_\lambda)\big)^2
\big)^{\frac{3}{2}}}
\nonumber
\\
&=\sup_{a\geq 0} 
\frac{2\delta \zeta_\lambda}{\pi} \frac{a+1}{a+h_\lambda}
\frac{\sin^2(\frac{|\lambda|\pi}{1-2|\lambda|})}{
1+\cos(\frac{|\lambda|\pi}{1-2|\lambda|})}
=\delta \cdot 
\frac{2 \zeta_\lambda}{h_\lambda \pi} 
\frac{\sin^2(\frac{|\lambda|\pi}{1-2|\lambda|})}{
1+\cos(\frac{|\lambda|\pi}{1-2|\lambda|})}
\label{Tfg2}
\end{align}
which becomes arbitrarily small for $\lambda\to 0$.

The contribution $\tau=1$ is more difficult, but can be controlled. 
Again we expect $F(t)$ to be of positive mean. We postpone the proof
and temporarily work with a conservative estimate 
$1+|\lambda|F(t)\geq h_\lambda:=1-\frac{|\lambda|}{5}$ for 
$0\leq |\lambda|\leq \frac{1}{6}$. Then 
$(a+1) \leq \frac{a+h_\lambda}{h_\lambda}$ and consequently 
\begin{align}
\|Tf-Tg\|^{(1)}_{LB}
&
\leq \sup_a \frac{\delta}{h_\lambda} \int_0^{\infty}\!dt \;
\frac{2|\lambda| 
\frac{\sin^3(\frac{|\lambda|\pi}{1-2|\lambda|})}{(|\lambda|\pi)^3}
(a+h_\lambda)\cdot  
(1+t)^{1-|\lambda|} \log(1+t) }{
\Big(t + (a{+}h_\lambda) 
\frac{\sin(\frac{\lambda|\pi}{1-2|\lambda|}) }{|\lambda|\pi}
\cos(\frac{\lambda|\pi}{1-2|\lambda|}) \Big)^3 }
\nonumber
\\
&=
\frac{2\delta |\lambda|}{h_\lambda \cos  (\frac{|\lambda|\pi}{1-2|\lambda|})}
\frac{\sin^2(\frac{|\lambda|\pi}{1-2|\lambda|})}{(|\lambda|\pi)^2}
\sup_{A_\lambda}\int_0^{\infty}\!dt \;
\frac{A_\lambda(a)\cdot  (1+t)^{1-|\lambda|} \log(1+t) }{
\big(t + A_\lambda(a)\big)^3 }\;,
\end{align}
where 
$A_\lambda(a):= (a{+}h_\lambda) 
\frac{\sin(\frac{\lambda|\pi}{1-2|\lambda|}) }{|\lambda|\pi}
\cos(\frac{\lambda|\pi}{1-2|\lambda|})$.
We use Young's inequality 
\begin{align}
\big(A_\lambda(a)\big)^\lambda(1+t)^{1-|\lambda|} &\leq \lambda A_\lambda(a) 
+(1-\lambda)(1+t)
\nonumber
\\
&=
(1-\lambda)+ (2\lambda-1) A_\lambda(a) + 
(1-\lambda)(t+ A_\lambda(a))
\label{Young}
\end{align}
to write 
\begin{align}
|Tf-Tg\|^{(2)}_{LB}
&
\leq \sup_{A_\lambda}
\frac{2\delta |\lambda|}{h_\lambda \cos  (\frac{|\lambda|\pi}{1-2|\lambda|})}
\frac{\sin^2(\frac{|\lambda|\pi}{1-2|\lambda|})}{(|\lambda|\pi)^2}
\big(A_\lambda(a)\big)^{1-|\lambda|}
\nonumber
\\*
&\times 
\int_0^{\infty}\!dt \;
\Big(
\frac{(1-\lambda)\log(1+t) }{(t + A_\lambda(a))^2 }
+\frac{((1-\lambda)+ (2\lambda-1) A_\lambda(a) )
\log(1+t) }{(t + A_\lambda(a))^3 }
\Big)
\nonumber
\\*
&= \delta \cdot 
\frac{(1-\frac{|\lambda|}{5})^{-1}}{
\cos  (\frac{|\lambda|\pi}{1-2|\lambda|})}
\frac{\sin^2(\frac{|\lambda|\pi}{1-2|\lambda|})}{(|\lambda|\pi)^2}
\cdot \sup_{A_\lambda} C_\lambda(A_\lambda(a))\;,
\end{align}
where (after integration by parts)
\begin{align}
C_\lambda(x)&:= |\lambda| x^{1-|\lambda|}
\int_0^{\infty}\!dt \;
\Big(
\frac{2(1-|\lambda|)}{(t+1)(t + x) }
+\frac{(1-|\lambda|)+ (2|\lambda|-1)x}{(1+t) (t + x)}
\Big)
\nonumber
\\
&=
\frac{-|\lambda|^2 + |\lambda|(1-2|\lambda|) (x-1) }{x^{|\lambda|}(x-1) }
+\frac{x^2 - (1-|\lambda|)x}{(x-1)^2} 
\frac{\log x^{|\lambda|}}{x^{|\lambda|}}
\;.
\end{align}
The maximum of $C_\lambda$ is governed by the function
$\frac{\log x^{|\lambda|}}{x^{|\lambda|}}$ which reaches $\frac{1}{e}$ 
at $x=e^{\frac{1}{|\lambda|}}$. 
For the range of $|\lambda|$ under consideration, this becomes huge so that all 
other terms except for $x^{|\lambda|}\approx e$ become 
negligible. Therefore we expect 
\[
\sup_x C_\lambda(x)
\leq  \frac{1+|\lambda|}{e}\;.
\]
A numerical investigation confirms this. 

\bigskip

It remains the contribution from $\tau=3$. There is a short cut
resulting from the crude bound $(\Delta R)^{(3)}(t) \leq \frac{\delta
  t}{|\lambda|}= \frac{(\Delta R)^{(3)}(t) }{|\lambda|^2\pi \zeta_\lambda}$.
Inserting this relation into (\ref{Tfg2}) gives
\begin{align}
\|Tf-Tg\|^{(3)}_{LB} \leq \frac{\delta}{1-\frac{|\lambda|}{5}} 
\frac{\sin^2(\frac{|\lambda|\pi}{1-2|\lambda|})}{(|\lambda|\pi)^2}
\frac{2}{1+\cos(\frac{|\lambda|\pi}{1-2|\lambda|})}\;.
\label{Tfg3-crude}
\end{align}
We show that this na\"{\i}ve bound is optimal. For that we start from 
Taylor's formula
\[
(\Delta R)^{(3)}(t)=
\delta |\lambda|
\int_0^1 d\xi \;\frac{(1-\xi) t
  (\log(1+t))^2}{(1+t)^{(1-\xi)|\lambda|}}\;.
\]
Up to an order $|\lambda|^2$-error we may replace $F(t)\mapsto
0$. Then 
\begin{align}
&\|Tf-Tg\|_{LB}^{(3)} 
\nonumber
\\
&=\sup_{a\geq 0} 
\int_0^1 \!\! d\xi \int_0^{\infty}\!dt \;
\frac{(1-\xi) t
 (\log(1{+}t))^2}{(1+t)^{(1-\xi)|\lambda|}}
\frac{2|\lambda|^2\delta  (1+a) }{
\big((|\lambda|\pi t)^2 
+ \big(a+1+|\lambda|\pi t\cot(\frac{|\lambda|\pi}{1-2|\lambda|}) 
\big)^2 \big)^{\frac{3}{2}}}
\nonumber
\\
&\leq 
\frac{\delta}{\cos (\frac{|\lambda|\pi}{1-2|\lambda|})}
\Big( \frac{\sin(\frac{|\lambda|\pi}{1-2|\lambda|})}{|\lambda|\pi}\Big)^2
\sup_{A_\lambda\geq 0} \tilde{C}_\lambda(A_\lambda)\;,
\label{Tfg3a}
\\
&\tilde{C}_\lambda(A_\lambda)
:=2|\lambda|^2
\int_0^1 d\xi \int_0^{\infty}\!dt \;
\frac{(1-\xi) A_\lambda (\log(1+t)^2) (1+t)^{1-(1-\xi)|\lambda|}
}{(t+A_\lambda)^3 }\;,
\nonumber
\end{align} 
where $A_\lambda:= \frac{(1+a)}{|\lambda|\pi}
\sin (\frac{|\lambda|\pi}{1-2|\lambda|})
\cos (\frac{|\lambda|\pi}{1-2|\lambda|})$.
Inserting Youngs's inequality 
(\ref{Young}) we get:
\begin{align}
\tilde{C}_\lambda (\alpha)
&=2|\lambda|^2
\int_0^1 \! d\xi \int_0^{\infty}\!dt \;
\Big\{
\nonumber
\\
& \qquad\qquad  
(1{-}\xi) \alpha^{1-(1-\xi)|\lambda|} (1-|\lambda|(1-\xi))
\Big(\frac{(\log(1+t))^2}{(t+\alpha)^2 }
+\frac{(\log(1+t))^2}{(t+\alpha)^3 }
\Big)
\nonumber
\\
& \qquad\qquad + 
(1-\xi) (2|\lambda|(1-\xi)-1)
\alpha^{2-(1-\xi)|\lambda|}  
\frac{(\log(1+t))^2}{
(t+\alpha)^3 }
\Big\}
\nonumber
\\
&= 2
\int_0^{\infty}\!dt \;
\Big\{
\alpha^{1-|\lambda|} \Big(
-\frac{2\alpha^{|\lambda|}-2}{(\log \alpha)^3}
+\frac{\alpha^{|\lambda|}+2|\lambda|-1}{(\log \alpha)^2}
-\frac{|\lambda|(1-|\lambda|)}{\log \alpha}\Big)
\nonumber
\\
& \qquad\qquad\qquad\qquad \times  \Big(\frac{(\log(1+t))^2}{(t+\alpha)^2 }
+\frac{(\log(1+t))^2}{(t+\alpha)^3 }
\Big)
\nonumber
\\
& \qquad+ 
\alpha^{2-|\lambda|}  
\Big(\frac{4\alpha^{|\lambda|}-4}{(\log \alpha)^3}
-\frac{\alpha^{|\lambda|}+4|\lambda|-1}{(\log \alpha)^2}
+\frac{|\lambda|(1-2|\lambda|)}{\log \alpha}\Big)
\frac{(\log(1+t))^2}{
(t+\alpha)^3 }
\Big\}.
\end{align}
We need the following integrals
\begin{align*}
&\int_0^\infty dt \;\frac{(\log(1+t))^2}{(t+\alpha)^2}
= \left\{
\begin{array}{c@{\text{\quad for }}l}
\dfrac{2\mathrm{Li}_2(1-\alpha)}{(1-\alpha)} 
& 0<\alpha <1 \\
2 & \alpha = 1
\\
\dfrac{(\log\alpha)^2 
+2  \mathrm{Li}_2(1-\frac{1}{\alpha})}{(\alpha-1)} & \alpha >1
\end{array}
\right.
\\
&\int_0^\infty dt \;\frac{(\log(1+t))^2}{(t+\alpha)^3}
= \left\{
\begin{array}{c@{\text{\quad for }}l}
\dfrac{-\log \alpha -
  \mathrm{Li}_2(1-\alpha)}{(1-\alpha)^2} 
& 0<\alpha <1 \\
\frac{1}{4} & \alpha = 1
\\
\dfrac{\frac{1}{2} (\log(\alpha))^2 -\log \alpha
+  \mathrm{Li}_2(1-\frac{1}{\alpha})}{(\alpha-1)^2} & \alpha >1
\end{array}
\right.
\end{align*}
We specify to $\alpha>1$ (the other cases are analytic continuations):
\begin{align}
\tilde{C}_\lambda (\alpha)
&= \frac{\alpha^2}{(\alpha-1)^2}
\Big\{
\Big(
1-\frac{1}{\alpha^{|\lambda|}}
-\frac{\log(\alpha^{|\lambda|})}{\alpha^{|\lambda|} }\Big)
\Big(1+2|\lambda|^2  \frac{\mathrm{Li}_2(1-\tfrac{1}{\alpha})}{
(\log\alpha^{|\lambda|})^2 }\Big)
\nonumber
\\
& \qquad\qquad\qquad
+ 
\Big(- \frac{8|\lambda|^2(1-\frac{1}{\alpha^{|\lambda|}})}{
(\log \alpha^{|\lambda|})^2 }
+ \frac{2|\lambda|(1-\frac{1-4|\lambda|}{\alpha^{|\lambda|}})}{(\log
    \alpha^{|\lambda|} )}
-\frac{2|\lambda|(1-2|\lambda|)}{\alpha^{|\lambda|} }\Big)
\Big\}
\nonumber
\\
&+ \frac{\alpha}{(\alpha-1)^2}\Big\{
\Big(
\frac{4|\lambda|^2(1-\frac{1}{\alpha^{|\lambda|}})}{
(\log \alpha^{|\lambda|})^2}
-\frac{2|\lambda|(1-\frac{1-2|\lambda|}{\alpha^{|\lambda|}}) }{
(\log \alpha^{|\lambda|} ) }
+\frac{2|\lambda|(1-|\lambda|)}{\alpha^{|\lambda|} }\Big)
\nonumber
\\
&
+ 
\Big(
\frac{2|\lambda|(1-\frac{1}{\alpha^{|\lambda|}})}{(\log \alpha^{|\lambda|})}
-1+\frac{1-2|\lambda|}{\alpha^{|\lambda|}}
+\frac{(\log \alpha^{|\lambda|}) (1-|\lambda|)}{
\alpha^{|\lambda|} }\Big)
\nonumber
\\
&\qquad\qquad\qquad\qquad\qquad\qquad\qquad\qquad
\times \Big(1+2|\lambda|^2  \frac{\mathrm{Li}_2(1-\tfrac{1}{\alpha})}{
(\log\alpha^{|\lambda|})^2 }\Big)
\Big\}\;.
\end{align}
This shows $\lim_{\alpha\to \infty} \tilde{C}_\lambda(\alpha)=1$. 
The next-to-leading terms turn out to be 
$1-\frac{\log x}{x}+\frac{2|\lambda|}{\log x}$, where $x:=
\alpha^{|\lambda|}$. This function gets bigger $1$ with a local
maximum $\approx 1+\frac{|\lambda|}{4}$ for $|\lambda|\leq
\frac{1}{6}$. 
A closer numerical simulation
confirms this bound  
$\sup_\alpha \tilde{C}_\lambda(\alpha)\leq 1+\frac{|\lambda|}{4}$ for all
$0\leq |\lambda|\leq \frac{1}{6}$. Inserted into (\ref{Tfg3a}) gives
no improvement compared with the crude bound (\ref{Tfg3-crude}).
\hfill $\square$%
\end{proof}

\section{Equicontinuity and Arzel\`a-Ascoli theorem}

The remaining task is to prove a variant of the Arzel\`a-Ascoli
theorem which establishes that if a
subset $\mathcal{T} \subseteq LB$ is equicontinuous and
pointwise bounded, then $\mathcal{T}$ is
compact. We start with the equicontinuity:

\begin{lemma}
The subset $T\mathcal{K}_\lambda \subseteq LB$ is equicontinuous in 
the norm topology of $LB$. 
More precisely, given $\epsilon>0$ one has 
$
\big|(1+a)(Tf)'(a)-(1+b)(Tf)'(b)\big| < \epsilon$ 
for all $f\in \mathcal{K}_\lambda$ and all $a,b\in \mathbb{R}_+$ 
with $|a-b|< \epsilon$.
\end{lemma}
\begin{proof}
We estimate via (\ref{Tf-prime})
\begin{align*}
&\big|(1+a)((Tf)'(a)-(1+b)(Tf)'(b))\big|
=\Big| \int_b^a dx \frac{d}{dx}
\big( (1+x)(Tf)'(x)\big)\Big|
\nonumber
\\
&=\Big| \int_b^a dx \;
\int_0^\infty dt\; \frac{d}{dx} \frac{|\lambda|(1+x)}{
(|\lambda|\pi t)^2 + (x+Rf(t))^2}
\Big|
\nonumber
\\
&=|\lambda|
\Big| \int_b^a dx \;
\int_0^\infty dt\; 
\frac{|\lambda|}{(|\lambda|\pi t)^2 + (x+Rf(t))^2 }
-\int_0^\infty dt\; 
\frac{
2|\lambda|(1+x)(x+Rf(t))}{\big(
(|\lambda|\pi t)^2 + (x+Rf(t))^2\big)^2 }
\Big|\;.
\end{align*}
We have the following upper bound:
\begin{align*}
\int_0^\infty dt\; \frac{2|\lambda|(1+x)(x+Rf(t))}{
((|\lambda|\pi t)^2 + (x+Rf(t))^2)^2}
&\leq 
\int_0^\infty dt\; \frac{2|\lambda|(1+x)}{
((|\lambda|\pi t)^2 + (x+1-\frac{|\lambda|}{5}+|\lambda|\pi t  
\cot \frac{|\lambda|\pi}{1-2|\lambda|})^2)^\frac{3}{2}}
\\*
&=\frac{2(1+x)\sin\frac{|\lambda|\pi}{1-2|\lambda|} }{\pi
(x+1-\frac{|\lambda|}{5})^2
(1+\cos \frac{|\lambda|\pi}{1-2|\lambda|})
}\;.
\end{align*}
We ignore possible cancellations and add the upper bound 
$\int_0^\infty dt\; \frac{|\lambda| }{
(|\lambda|\pi t)^2 + (x+Rf(t))^2 }
\leq \frac{|\lambda|}{1-2|\lambda|} \frac{1}{1+x}$ established 
in the proof of $T\mathcal{K}_\lambda\subseteq 
\mathcal{K}_\lambda$. Taking also the supremum in $x$ we conclude
\begin{align*}
&\big|(1+a)((Tf)'(a)-(1+b)(Tf)'(b))\big| 
\\
&\leq |a-b|\frac{|\lambda|}{1-2|\lambda|} 
\cdot \Big(1 + \frac{\sin\frac{|\lambda|\pi}{1-2|\lambda|} }{
\frac{|\lambda|\pi}{1-2|\lambda|} } \frac{2(1-\frac{|\lambda|}{5})^{-2}}{
(1+\cos \frac{|\lambda|\pi}{1-2|\lambda|})}\Big)\;.
\end{align*}
The rhs is $\leq |a-b|$ for any $0\leq |\lambda|\leq \frac{1}{6}$.
\hfill $\square$%
\end{proof}

The standard Arzel\'a-Ascoli theorem concerns continuous functions on
\emph{compact} spaces. This can largely be generalised to
$\mathcal{C}(X,Y)$ equipped with the compact-open topology relative to
general Hausdorff spaces $X,Y$, see \cite{Myers:1946??}. The idea is to
prove that for an equicontinuous family $\mathcal{T}$, 
the compact-open topology and
the pointwise topology coincide. Pointwise compactness of
$\mathcal{T}(x)$ for every $x\in  X$ implies 
compactness of $\prod_{x\in X} \mathcal{T}(x)$ by Tychonoff's theorem,
thus compactness of the equicontinuous family $\mathcal{T}$ in
the compact-open topology. We cannot make use of this setting because
to prove continuity of $T$ we had to control the Hilbert transform via
the global behaviour of functions in $\mathcal{K}_\lambda$. It seems unlikely
that this can be replaced by a local control in the compact-open
topology.

Being forced to work in norm topology, the only chance to rescue
Arzel\'a-Ascoli for equicontinuous families in $LB$ is to restrict to
compact subsets of $\mathbb{R}_+$. This is not unreasonable because we
worked originally over the cut-off space $[0,\Lambda^2]$. We find it
necessary to  reprove the Arzel\'a-Ascoli theorem for equicontinuous
subsets of $LB$.

\begin{lemma}
The subset $T\mathcal{K}_\lambda\subseteq LB$ is relatively compact in the 
$\|~\|_{LB}$ topology if restricted to any compact interval 
$[0,\Lambda^2]$.
\end{lemma}
\begin{proof}  Choose any $\Lambda^2>0$. The family
$T\mathcal{K}_\lambda\subseteq LB$ is bounded and equicontinuous on
$[0,\Lambda^2]$ with respect to $f\mapsto (1+x)f'(x)$. On metric
spaces such as $LB$, compactness is equivalent to sequentially
compactness. We thus have to prove that any sequence $(f_k) \in
T\mathcal{K}_\lambda$ has a $\|~\|_{LB}$-convergent subsequence when
restricted to $[0,\Lambda^2]$.

Given $\epsilon>0$, there is for every $0<x<\Lambda^2$ an open
$\frac{\epsilon}{3}$-neighbourhood
$U_{\frac{\epsilon}{3}}(x):=\{y\in \mathbb{R}_+\;:~
  |y-x|<\frac{\epsilon}{3}\}$ which by the equicontinuity of
  $T\mathcal{K}_\lambda$ has the property that
\[
\big|(1+s)f'(s)- (1+x)f'(x)\big| < \frac{\epsilon}{3} \quad
\text{for all } s\in U_{\frac{\epsilon}{3}}(x) \text{ and all } 
f\in T\mathcal{K}_\lambda\;.
\]
These $\big\{U_{\frac{\epsilon}{3}}(x)\big\}_{0<x<\Lambda^2}$ form
an open cover of $[0,\Lambda^2]$ which by the compactness of 
$[0,\Lambda^2]$ can be reduced to a finite subcover 
$\big\{U_{\frac{\epsilon}{3}}(x_i)\big\}_{i=1,\dots,N}$
(it is this step which does not work for $\mathbb{R}_+$). It suffices to take 
$x_i=\frac{\epsilon}{4}(2i-1)$ and thus
$N=\frac{2\Lambda^2}{\epsilon}$. 

Start at $x_1$ and note that $((1+x_1)f'_k(x_1))_{k\in\mathbb{N}}$ is 
bounded for every member of the sequence $(f_k)$. 
By the Bolzano-Weierstra\ss{} theorem there is a subsequence
$(f_{k_1}(x_1))_{k_1 \in \mathbb{N}}$ such that
$((1+x_1)f'_{k_1}(x_1))_{k_1\in\mathbb{N}}$ converges at $x_1$.
Repeat this to construct a subsequence $(f_{k_2})_{k_2 \in
  \mathbb{N}}$ of $(f_{k_1})_{k_1 \in \mathbb{N}}$ such that 
both $((1+x_1) f_{k_2}'(x_1))_{k_2\in \mathbb{N}}$ and 
$((1+x_2) f_{k_2}'(x_2))_{k_2\in \mathbb{N}}$ converge. And so
on. This eventually produces a subsequence 
$(f_{k_N})_{k_N\in \mathbb{N}}$ of $(f_k)$ which has the property
that $((1+x_i) f_{k_N}'(x_i))_{k_N\in \mathbb{N}}$ converges for every
$i=1,\dots N$. We rename $(f_{k_N})_{k_N\in \mathbb{N}}=(\tilde{f}_\ell)_{\ell \in
  \mathbb{N}}$ for simplicity. 

Convergence implies that for every 
$i=1,\dots ,N$ there is a $K_i(\epsilon) \in \mathbb{N}$ such that 
\[
\big| (1+x_i)\tilde{f}_\ell'(x_i)- (1+x_i)\tilde{f}'_m(x_i)\big|
<\frac{\epsilon}{3} \quad \text{for all } 
\ell,m\geq K_i(\epsilon)\;.
\]

Given any $x\in [0,\Lambda]$, choose one index $j\in \{1,\dots,N\}$
such that $x\in U_{\frac{\epsilon}{3}}(x_j)$. 
Then for any 
$\ell,m \geq K(\epsilon):=\max_{i=1,\dots,N} K_i(\epsilon)$ one has
\begin{align*}
\big|(1+x)\tilde{f}_\ell'(x)-
(1+x)\tilde{f}_m'(x)\big| 
&< \big|(1+x)\tilde{f}_\ell'(x)-
(1+x_j)\tilde{f}_\ell'(x_j)\big| 
\\
&+\big|(1+x_j)\tilde{f}_\ell'(x_j)-
(1+x_j)\tilde{f}_m'(x_j)\big| 
\\
& +
\big|(1+x_j)\tilde{f}_m'(x_j)-
(1+x)\tilde{f}_m'(x)\big| <\epsilon
\end{align*}
In other words, any sequence $(f_k)_{k\in \mathbb{N}}$ in
$T\mathcal{K}_\lambda$ has a subsequence $\tilde{f}_{\ell \in \mathbb{N}}$
such that 
$\big((1+x)\tilde{f}_\ell'(x)\big)_{\ell \in  \mathbb{N}}$ converges
uniformly on any compact interval $[0,\Lambda^2]$ to a differentiable 
limit function which belongs to the closure $\overline{T\mathcal{K}_\lambda}
\subseteq \mathcal{K}_\lambda$. This means that 
$T\mathcal{K}_\lambda$ is $\|~\|_{LB}$-relatively compact in $LB$ if
restricted to $[0,\Lambda^2]$.
\hfill $\square$%
\end{proof}

\section{Conclusions}

In proving existence of a solution of (\ref{G0b})
we closed a major gap in our programme to construct a solvable
quantum field theory model in four dimensions. In 
\cite{Grosse:2014lxa} we have studied the numerical iteration of 
(\ref{G0b}) in the spirit of the Banach fixed point theorem and
convinced ourselves that the iteration converges numerically. As shown
in Figure~\ref{fig:2} 
\begin{figure}[t]
\begin{picture}(120,94)
\put(0,46){\includegraphics[width=5.8cm,bb= 0 0 260 160]{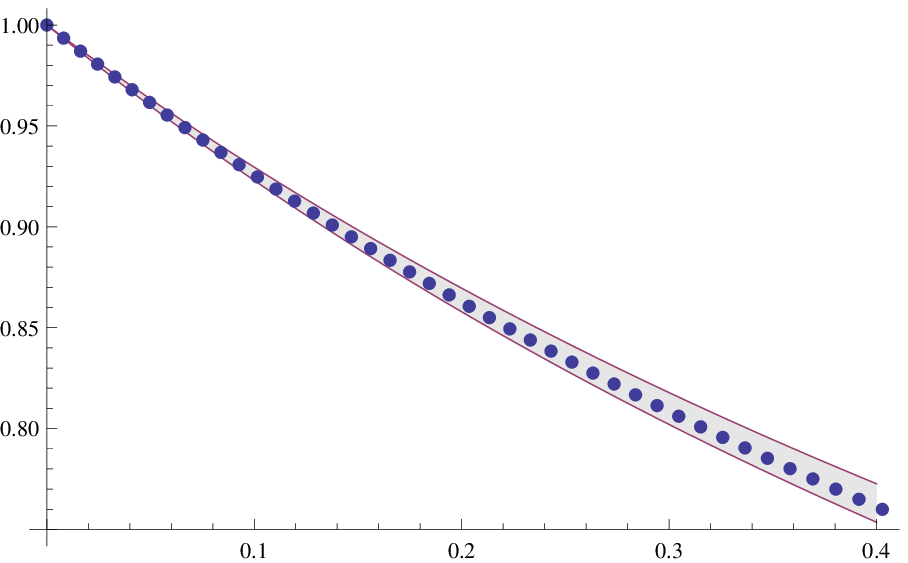}}
\put(74,46){\includegraphics[width=5.8cm,bb= 0 0 260 162]{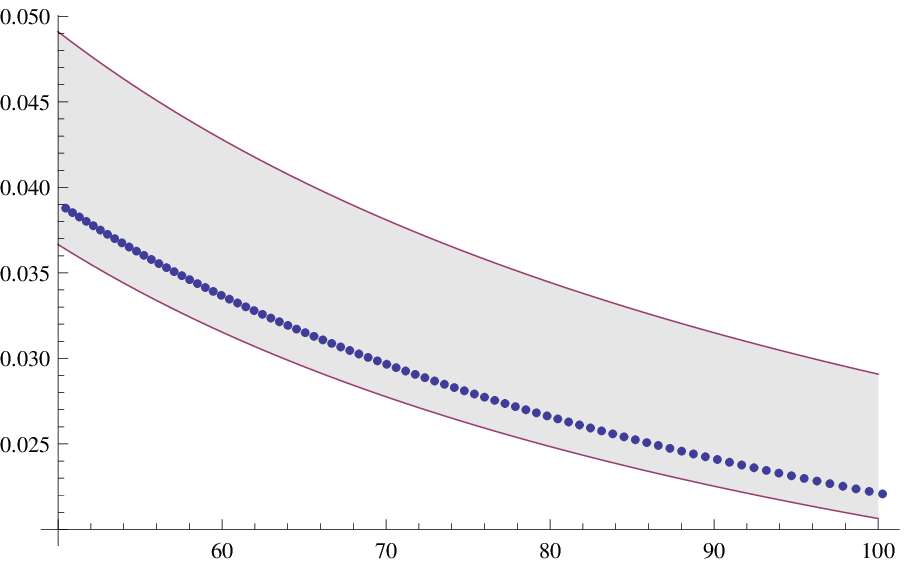}}
\put(0,0){\includegraphics[width=5.8cm,bb= 0 0 260 156]{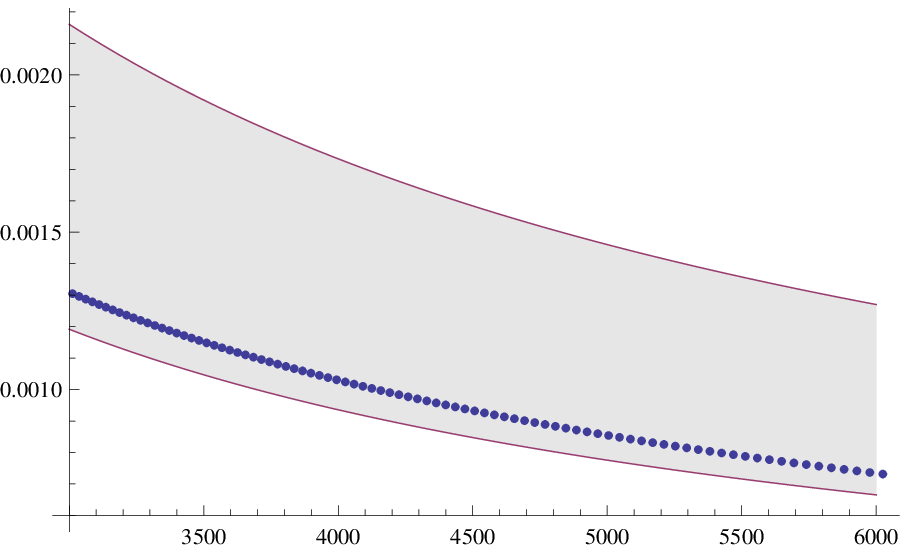}}
\put(74,00){\includegraphics[width=5.8cm,bb= 0 0 260 151]{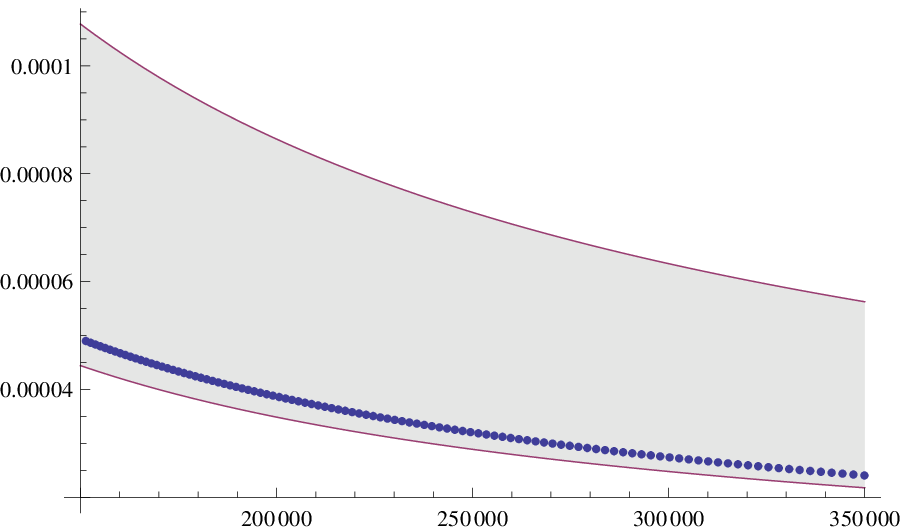}}
\end{picture}
\caption{\label{fig:2} 
Comparison between the numerical solution
$b\mapsto G_{0b}$ (obtained in \cite{Grosse:2014lxa})
of the equation 
  (\ref{G0b}) for $\lambda=-\frac{1}{2\pi}$ (blue dots) with the
  domain $\exp \mathcal{K}_\lambda$ (shaded region, defined in
  (\ref{calK})) in which we proved
  existence of a fixed point. Observe the big variation of
  $b$-intervals and corresponding values $G_{0b}$.}
\end{figure}
there is \emph{perfect agreement} between the 
numerical solution (at $\lambda=-\frac{1}{2\pi}$) and the 
analytically established fixed point domain $\exp (\mathcal{K}_\lambda)$.

The numerical treatment \cite{Grosse:2014lxa} leaves no doubt that the
solution $G_{0b}$ inside $\exp (\mathcal{K}_\lambda)$ is unique.  It
would be very desirable to prove this also analytically. As shown in
the appendix where we prove that also $G_{0b}=1$ solves (\ref{G0b})
for $\lambda<0$, the restriction to $\exp (\mathcal{K}_\lambda)$ is
essential.  We slightly missed in Prop.~\ref{prop:continuity} the
contractivity criterion of the Banach fixed point theorem. If we knew
the asymptotic exponent $\lim_{b\to \infty} \frac{-\log G_{0b}}{\log
  (1+b)}$ then we could considerably improve the bound (\ref{zeta-0})
by an integration from the other end. Another strategy would be to
prove that, starting with the very good estimate $f^{(0)}(b):=\log
G_{0b}^{(0)}=-(1-|\lambda|) \log(1+b)$, one has $(Tf^{(n)})(b)=:
f^{(n+1)}(b)\geq f^{(n)}$. Together with the boundedness proved here,
such a monotonicity would also imply uniqueness.

As discussed in \cite{Grosse:2013iva} and \cite{Grosse:2014lxa} it is
very important to know that $G_{0b}$ is a Stieltjes function (see
e.g.\ \cite{Berg}). We have no doubt that this is true, but the proof
is missing. The boundaries of $\exp (\mathcal{K}_\lambda)$ are
Stieltjes and the numerical solution is parallel to these boundaries
(Figure~\ref{fig:2}). We made recently some progress in this direction
using results of this paper in an essental way: We can prove that
\emph{any} fixed point solution $G_{0b}$ of (\ref{G0b}) inside $\exp
(\mathcal{K}_\lambda)$ has a holomorphic continuation $z\mapsto
G_{0z}$ to complex $z$ with $\mathrm{Re}(z)>-1+\frac{|\lambda|}{5}$
(in fact a bit more) and satisfies the anti-Herglotz property
$\mathrm{Im}(G_{0z})\leq 0$ for $\mathrm{Im}(z)>0$ in that half space.
To prove the Stieltjes property we have to extend these results to the
cut plane $\mathbb{C}\setminus {]-\infty,0]}$, see \cite{Berg}. The
estimates proved in this paper will definitely be relevant for this
step.

\begin{appendix}
\renewcommand{\theequation}{\thesection.\arabic{equation}}
\setcounter{equation}{0}
\makeatletter\@addtoreset{equation}{section}\makeatother

\section{The fixed point operator applied 
to the constant function}

\label{sec:G0b=1}

We have proved in sec.~\ref{sec:TKK} that the operator $T$ defined in
(\ref{Tfa}) maps $\mathcal{K}_\lambda$ defined in (\ref{calK}) into
itself. We add a small note showing the existence of fixed
points outside $\mathcal{K}_\lambda$. Concretely we show that $T0$
converges pointwise to $0$ for $\Lambda^2\to \infty$. We have to
reintroduce a finite cut-off $\Lambda^2$ to make sense of the Hilbert
transform of $\exp(0)=1$, namely
$\mathcal{H}^{\!\Lambda^2}_p(1)=\frac{1}{\pi} \log
\frac{\Lambda^2-p}{p}$. We then have for (\ref{Tf-prime})
\begin{align}
(T0)'(b) 
&:= -\frac{1}{1+b} + |\lambda| 
\int_0^{\Lambda^2} \frac{dp}{
(|\lambda| \pi p)^2 + 
(b+1-|\lambda|p\log
\frac{\Lambda^2-p}{p})^2}
\nonumber
\\
&= -\frac{1}{1+b} +\frac{1}{|\lambda|\Lambda^2}
\int_0^{\infty} \frac{dq}{
\pi^2 + \big(\frac{1+b}{\Lambda^2|\lambda|}(1+q)
-\log q\big)^2}\;, \label{T1-}
\end{align}
where we have substituted $\frac{\Lambda^2-p}{p}=q$. We prove:
\begin{lemma} For $u>0$ one has
$\displaystyle 
\int_0^{\infty} \frac{dq}{
\pi^2 + \big(u(1+q)-\log q\big)^2} 
= \frac{1}{u(u+1)}$.
\label{Lemma:Cauchy-}
\end{lemma}
\begin{proof} We have
\begin{align}
\int_0^\infty \frac{dq}{\pi^2+\big(u(1+q)-\log q\big)^2}
&= \int_0^\infty \frac{dq}{2\pi \mathrm{i}} 
\Big( \Big\{\frac{1}{ -u(1+q)+\log q-\mathrm{i}\pi}+\frac{1}{u(1+q)}\Big\}
\nonumber
\\*
&- \Big\{\frac{1}{ -u(1+q)+\log q+\mathrm{i}\pi}+\frac{1}{u(1+q)}\Big\}\Big)\;.
\end{align}
The terms $\frac{1}{u(1+q)}$ are added to improve the deacy at
infinity. We put $z=qe^{\mathrm{i}\epsilon}$ in the first $\{\dots\}$ 
and $z=qe^{\mathrm{i}(2\pi -\epsilon)}$ in the second $\{\dots\}$. Then for 
$\epsilon\to 0$ we have
\begin{align*}
\begin{array}{l}
\displaystyle
\pm \int_0^\infty \frac{dq}{2\pi \mathrm{i}} 
\Big\{\frac{1}{ -u(1+q)+\log
  q \mp\mathrm{i}\pi}+\frac{1}{u(1+q)}\Big\}
\\[2ex]
= \displaystyle \lim_{\epsilon\to 0}\int_{c_\pm}
\frac{dz}{2\pi \mathrm{i}} 
\Big\{\frac{1}{ -u(1+z)+\log
  (ze^{-\mathrm{i}\pi})}+\frac{1}{u(1+z)}\Big\}
\end{array}\quad\qquad 
\parbox[c]{30mm}{\unitlength.5mm\begin{picture}(50,40)
\put(0,25){\vector(1,0){50}}
\put(25,0){\vector(0,1){50}}
\thicklines
\put(25,25){\vector(4,1){12}}
\put(25,25){\line(4,1){24}}
\put(25,25){\line(4,-1){24}}
\put(49,19){\vector(-4,1){12}}
\put(35,31){\oval(28,28)[tr]}
\put(35,19){\oval(28,28)[br]}
\put(18,31){\oval(28,28)[tl]}
\put(18,19){\oval(28,28)[bl]}
\put(15,45){\line(1,0){20}}
\put(15,5){\line(1,0){20}}
\put(4,15){\line(0,1){20}}
\put(37,16){\mbox{\small$c_-$}}
\put(37,32){\mbox{\small$c_+$}}
\put(10,40){\mbox{\small$c_\infty$}}
\put(18,25){\makebox(0,0){\mbox{$\times$}}}
\put(18,21){\makebox(0,0){\mbox{\small$-1$}}}
\end{picture}}
\end{align*}
with $\mathbb{R}_+$ chosen as the cut of $\log (ze^{-\mathrm{i}\pi})$. The decay at
$\infty$ guarantees that the intgral over the arc $c_\infty$ does not
contribute. Therefore the residue theorem gives
\begin{align}
\int_0^\infty \frac{dq}{\pi^2+\big(u(1+q)+\log q\big)^2}
&=\sum_{z\in \mathbb{C}\setminus \mathbb{R}_+}
\mathrm{Res}
\Big( \frac{1}{ -u(1+z)+\log (ze^{-\mathrm{i}\pi})}
+\frac{1}{u(1+z)}\Big)\;.
\end{align} 
For $z=|z|e^{\mathrm{i}\phi}$ with $0<\phi<\pi$ one has 
$\mathrm{Im}(-u(1+z)+\log (ze^{-\mathrm{i}\pi}))=
-u|z|\sin \phi-(\pi-\phi)<0$. Therefore,
the residue equation $0= u(1+z)+\log (ze^{-\mathrm{i}\pi})$ has
solutions only on the negative real axis: $z=-x$ and 
$u(1-x)=\log x$ with unique solution $x=1$. This gives
\begin{align}
\int_0^\infty \frac{dq}{\pi^2+\big(u(1+q)-\log q\big)^2}
&=
\Big( \frac{1}{ -u+ \frac{1}{z}}\Big|_{z=-1} + \frac{1}{u}\Big)
= \frac{1}{u(u+1)}\;.
\tag*{\mbox{$\square$}}
\end{align} 
\end{proof}
Insertion into (\ref{T1-}) gives 
\begin{align}
(T0)'(b)=-\frac{1}{|\lambda|\Lambda^2+1+b}
\quad\Rightarrow\quad 
(T0)(b)=\log \Big(\frac{1}{1+\frac{b}{1+|\lambda|\Lambda^2}}\Big)\;,
\end{align}
which is pointwise convergent to $0$ for $\Lambda^2\to \infty$. 
This means that $G_{0b}=\exp(0)=1$ for all $b$ is a solution of 
(\ref{G0b}) for $\lambda<0$.

This solution is interesting in so far as the numerical investigation
in \cite{Grosse:2014lxa} shows a phase transition at critical 
coupling constant $\lambda_c\approx -0.39$. For $\lambda_c<\lambda\leq
0$ we find qualitative agreement with $\exp(\mathcal {K}_\lambda)$, 
see Figure~\ref{fig:2}, 
whereas
for $\lambda<\lambda_c$ we have $G_{0b}=1$ in a whole neighbourhood of
$b=0$. This suggests that $\lambda_c$ locates the transition between
solutions $G_{0b}\in \exp(\mathcal {K}_\lambda)$ and $G_{0b}=\exp(0)=1$.

\end{appendix}

\end{document}